\documentclass[12pt]{article}
\usepackage{authblk}

\usepackage{latexsym,amssymb,mathtools,mathabx}
\usepackage{graphicx}

\usepackage{amsmath,amsthm}
\usepackage{algpseudocode}
\usepackage[]{todonotes}
\usepackage{tikz-qtree}
\usepackage{url}
\usepackage{varwidth}

\usepackage[colorlinks=true, linkcolor=blue]{hyperref}%

\usepackage{xcolor}
\definecolor{blue2}{HTML}{5DADE2}

\usepackage{caption, subcaption}
\usepackage{multirow}
\usepackage{multicol}
\usepackage{rotating}
\usepackage{tabto}    
\usepackage{array}
\usepackage{booktabs}
\usepackage{subfiles}
\usepackage{pgfplots}
\usepgfplotslibrary{external}
\usepackage{tcolorbox}

\usetikzlibrary{arrows,shapes,automata,backgrounds,petri,positioning,decorations.markings}
\renewcommand\leq\varleq
\renewcommand\leqslant\le
\renewcommand\geq\vargeq
\renewcommand\geqslant\ge

\newcommand{\bpm}{\begin{pmatrix}}
\newcommand{\epm}{\end{pmatrix}}

\theoremstyle{definition}
\theoremstyle{plain} 

\newtheorem{ass}{Assumption}[section]

\newtheorem{lemma}[ass]{Lemma}

\newtheorem{problem}{Test problem}

\usepackage[ruled,noend,vlined,linesnumbered]{algorithm2e} 

\newcommand{\True}{\textsc{true}}
\newcommand{\False}{\textsc{false}}
\makeatletter
\let\original@algocf@latexcaption\algocf@latexcaption
\long\def\algocf@latexcaption#1[#2]{%
  \@ifundefined{NR@gettitle}{%
    \def\@currentlabelname{#2}%
  }{%
    \NR@gettitle{#2}%
  }%
  \original@algocf@latexcaption{#1}[{#2}]%
}
\makeatother
\makeatletter
\newcommand*{\algotitle}[2]{%
  \stepcounter{algocf}%
  \hypertarget{algocf.title.\theHalgocf}{}%
  \NR@gettitle{#1}%
  \label{#2}%
  \addtocounter{algocf}{-1}%
}
\makeatother

\usepackage{algpseudocode}

\algnewcommand{\algorithmicand}{\textbf{ and }}
\algnewcommand{\algorithmicor}{\textbf{ or }}
\algnewcommand{\algorithmicnot}{\textbf{ not }}
\algnewcommand{\algorithmicis}{\textbf{ is }}
\algnewcommand{\algorithmicnull}{\text{ null }}
\algnewcommand{\OR}{\algorithmicor}
\algnewcommand{\AND}{\algorithmicand}
\algnewcommand{\NOT}{\algorithmicnot}
\algnewcommand{\IS}{\algorithmicis}
\algnewcommand{\NULL}{\algorithmicnull}

\DeclareMathOperator*{\argmin}{arg\,min}

\title{A generic Branch-and-Cut algorithm for bi-objective binary linear programs}
\author[1]{Pierre Fouilhoux, Lucas L\'etocart, Yue Zhang}
\affil[1]{LIPN - Université Sorbonne Paris Nord, Villetaneuse - CNRS\\ 
\small \{pierre.fouilhoux,lucas.letocart,yue.zhang@lipn.univ-paris13.fr\}
}
\date{\today}

\begin{document}

\maketitle

{\small\bf Abstract:}
This paper presents the first generic bi-objective binary linear branch-and-cut algorithm.  Studying the impact of valid inequalities in solution and objective spaces, two cutting frameworks are proposed. The multi-point separation problem is introduced together with a cutting algorithm to efficiently generate valid inequalities violating multiple points simultaneously. The other main idea is to invoke state-of-the-art integer linear programming solver's internal advanced techniques such as cut separators. Aggregation techniques are proposed to use these frameworks with a trade-off among efficient cut separations, tight lower and upper bound sets and advanced branching strategies. Experiments on various types of instances in the literature exhibit the promising efficiency of the algorithm that solves instances with up to 2800 binary variables in less than one hour of CPU time. Our algorithms are easy to extend for more than two objectives and integer variables.



\section{Introduction}

Many real-world combinatorial optimization applications are characterized by multiple, usually conflicting objectives. Enumerating the complete set of \textit{Pareto optimal} (so-called efficient) solutions remains computationally challenging, since the Pareto optimal set usually is not polynomial with respect to the size of the problem instance. Our ambition is to extend the Branch-and-Bound (B\&B) framework developed in efficient optimization solvers to generic multi-objective problems, and embed with efficient valid cuts. 

In this work, we design and implement a Branch-and-Cut (B\&C) method for Bi-Objective Binary Linear Problems (BOBLP). Note that this can be easily generalized to Bi-Objective Integer Linear Problems (BOILP) as well. The BOBLP is written as $\{ \min z(x) = (z_1(x), \,  z_2(x)) \text{ s.t. } x \in \mathcal{X} = \{x \in \{0,1\}^n : Ax \leqslant b\} \}$, where $z_k(x) = c_k x \, , c_k \in \mathbb{R}^n \, \forall k = [1,2]$, $A \in \mathbb{R}^{m \times n}$, $b \in \mathbb{R}^m$. The set of two linear functions $z(\cdot) : \mathbb{R}^n \rightarrow \mathbb{R}^2$ maps a solution from the \textit{solution space} $\mathbb{R}^n$ to the \textit{criteria (or objective) space} $\mathbb{R}^2$. Given a solution $x$ of the \textit{feasible set} $\mathcal{X}$, $y = z(x)$ is a \textit{point} of $\mathcal{Y} = z(\mathcal{X})$.

Given two different points $y$ and $y'$ in $\mathbb{R}^2$, $y \leq y'$ ($y$ \textit{dominates} $y'$) if $y_1 \leqslant y'_1$ and $y_2 \leqslant y'_2$. A feasible solution $x \in \mathcal{X}$ is called \textit{efficient} if no other feasible solution dominates itself. Let $\mathcal{X}_E = \{ x \in \mathcal{X} \, | \, \nexists x'\in \mathcal{X}, z(x') \leq z(x)\}$ denote the \textit{efficient set}. The BOBLP is to provide $\mathcal{X}_E$. The image of an efficient solution in criteria space is called \textit{non-dominated point} and $\mathcal{Y}_N = z(\mathcal{X}_E)$ is the \textit{non-dominated set}. In this work, we propose an exact Bi-Objective Binary Linear Branch-and-Cut (BOBLB\&C) algorithm to solve BOBLP.

In the literature, many exact algorithms cope with BOILP problems by solving a series of single-objective Integer Linear 
 Problems (ILP) with the aid of powerful solvers, which are known as the criteria space search algorithm category. The $\epsilon$-constraint algorithm \cite{chankong1983, mavrotas2009effective} is the state-of-the-art method that optimizes repeatedly in one objective direction with a moving step $\epsilon \in \mathbb{R}$ in the other objective. However this method provides $\mathcal{Y}_N$ but not all $\mathcal{X}_E$, unless each ILP run enumerates all the optimal solutions. The dichotomy algorithm \cite{aneja1979bicriteria} is another classical approach which not provides all $\mathcal{Y}_N$ but only the (supported) extreme non-dominated points of $conv (\mathcal{Y})_N$. The balanced box search algorithm proposed by \cite{boland2015criterion} exploits $\mathcal{Y}_N$ with a complexity worse than the one of $\epsilon$-constraint but without the numerical issues raised by $\epsilon$ value. The weighted Tchebycheff scalarization and its varieties algorithms \cite{10.1007/978-3-642-87563-2_5, ralphs2006improved, steuer1983interactive} produce $conv (\mathcal{Y})_N$. 
 Other related scalarization algorithms are reviewed in \cite{ehrgott2006discussion}.

Note that not all of the previously mentioned algorithms provide $\mathcal{X}_E$ nor even $\mathcal{Y}_N$, or can be directly applied in many-objective context. Moreover, solving the ILP problem is generally $\mathcal{NP}$-hard, and calling the ILP solver may highly time consuming for large or hard problems. The other algorithm category called solution space search mainly refers to multi-objective B\&B(\&C) algorithms that enumerate all efficient solutions by solving the relatively simple and relaxed subproblems divided from the solution space.

The extension of the classic single-objective B\&B for bi-objective \cite{kiziltan1983algorithm} and many-objective problems \cite{forget2022warm} has arisen in the state-of-the-art approaches. Most of the BOB\&B algorithms are investigated for particular problems, such as the bi-objective minimum spanning tree problem (\cite{ramos1998problem}, \cite{sourd2008multiobjective}), the bi-objective multi-dimensional knapsack problem (\cite{cerqueus2015bi}, \cite{florios2010solving}). The latest advances in the Multi-Objective Branch-and-Bound (MOB\&B) algorithm can be found in this survey \cite{przybylski2017multi}. In the recent decade, significant developments of BOB\&C algorithms have been studied in \cite{jozefowiez2012generic, adelgren2022branch, gadegaard2019bi}. The idea is to invoke the single-objective cutting planes iteratively on each lower bound, and the efficiency of cutting improvements is preliminarily shown by applying valid cuts dedicated to specific problems. Further improvements of the classical B\&B framework are to explore and to decompose both the criteria and solution space, such as slicing \cite{stidsen2018hybrid}, Pareto branching \cite{stidsen2014branch}, extended Pareto branching \cite{gadegaard2019bi} and with three objectives \cite{gadegaard2020branch}. 

With the belief that strengthening the relaxation bounds would help the nodes fathoming, we design a generic BOBLB\&C algorithm with different cutting approaches for binary linear programs but applicable to integer problems as well. In Section \ref{sec:BOBB}, we introduce the preliminary notions and the Bi-Objective Binary Linear Branch \& Bound (BOBLB\&B) framework implemented from scratch. In Section \ref{sec:BOBC}, we exhibit two efficient generic cutting approaches, which are the multi-point and the ILP solver cutting planes. Additionally, we balance the lower bound set exactness with running time. A detailed computational analysis of our algorithms' performance is given in Section \ref{sec:experiments}. Finally, this paper ends with the concluding remarks and perspectives in Section \ref{sec:conclusion}.


\section{Bi-Objective Branch-and-Bound algorithm}\label{sec:BOBB}

The BOBLB\&B method presented in this section is based on classical Branch\&Bound approaches for mono-objective problems but generalizes different aspects such as bounds, branching and pruning.

\subsection{Ingredients of BOBLB\&B algorithm}

\subsubsection{Maintenance of bound sets}

We consider the lower and upper bound sets, respectively denoted $\mathcal{L}$ and $\mathcal{U}$, defined by \cite{ehrgott2007bound}. Denote $\mathcal{N}(\mathcal{U})$ the set of \textit{local nadir points} of $\mathcal{U}$. whose coordinates are $ (u^{''}_1 ,u^{'}_2) $ for every two consecutive points $u', u''$  of $\mathcal{U}$ from the left to the right in objective space.

\paragraph{Local relaxed bound set}

The integrity relaxation BOLP of BOBLP is $\{\min z(x) = (z_1(x), \, z_2(x)) \text{ s.t. } x \in \widetilde{\mathcal{X}} = \{x \in [0,1]^n : Ax \leqslant b\}  \}$. Note that the non-dominated set of BOLP is the non-dominated boundaries of the polyhedron $\widetilde{\mathcal{Y}} = z(\widetilde{\mathcal{X}} )$, denoted $\widetilde{\mathcal{Y}}_N$.  \cite{ehrgott2007bound} have shown that $ \widetilde{\mathcal{Y}}_N$ is a valid lower bound set. The line segments of $ \widetilde{\mathcal{Y}}_N$ are linked by the extreme supported non-dominated relaxed points that can be exhaustively enumerated by solving the weighted sum problem $\widetilde{P}(\lambda) \{\min_{x \in \widetilde{\mathcal{X}}}  \sum_{k=1}^2 \lambda_k z_k(x) \}$, where $\lambda$ are the scalar parameters in dichotomy algorithm.

\paragraph{Global upper bound set} During the BOBLB\&B tree search process, we keep the set of feasible non-dominated points encountered so far, which is filtered by dominance. In \autoref{algo:update_incumbent}, $\mathcal{U}$ contains the non-dominated feasible points after adding $y$ the new integer point just found. Note that we store in $\mathcal{S}$ the equivalent solutions mapping the same point in criteria space. At the end of the BOBLB\&B tree search, the global upper bound set $\mathcal{U}$ is $\mathcal{Y}_N$ and $\mathcal{S}(\mathcal{U})$ equals to $\mathcal{X}_E$. In our implementation, the upper bound set can be initialized with either a set of heuristic solutions or an empty set. 

\begin{algorithm}
\scriptsize
\KwIn{$\mathcal{U}$ an upper bound set, $\mathcal{S}$ the dictionary associating the list of equivalent solutions with a criteria point, and $y$ the new criteria point found.}
\KwOut{$\mathcal{U}$ the non-dominated feasible points updated with filtering dominance.}

\For{$y' \in \mathcal{U}$}{
\If{$y' = y$ }{
    \For{$x \in S(y)$}{$S(y').add(x)$ \;}
    \Return $\mathcal{U}$ \;
}\ElseIf{$y' \leq y $ }{
   \Return $\mathcal{U}$  \tcp*{$y$ is dominated by $\mathcal{U}$}
}\ElseIf{$y \leq y'$}{
    $\mathcal{U}.remove(y')$ \;
}
}
$\mathcal{U}.add(y)$ \;
\Return $\mathcal{U}$ \;
\caption{\textsc{Update upper bound set}($\mathcal{U}$, $\mathcal{S}$, $y$)} 
\label{algo:update_incumbent}
\end{algorithm}

\subsubsection{Three rules of node fathoming}

Analogously to the single-objective B\&B algorithm, the fathoming rules also consist of three basic cases for bi-objective problems.

\paragraph{Fathoming by infeasibility} Given a subproblem, if the feasible region defined by the linear system with the partial assignment on variables is empty, then this node is pruned by infeasibility.

\paragraph{Fathoming by integrity} A bi-objective subproblem can be pruned by integrity if all the non-dominated integer points for the current subproblem are known. The current node is pruned when the BOLP subproblem has a unique extreme point which is integer.
        
\paragraph{Fathoming by dominance} By \cite{forget2022warm}, if the local lower bound set is dominated by the upper bound set, i.e. if for each local nadir point $u \in \mathcal{N}(\mathcal{U})$, $u \notin \mathcal{L} + \mathbb{R}^2_{\geqq} $, then we can conclude that all the non-dominated points of the current subproblem will be dominated by the upper bound set. The fully explicit dominance test, given in \autoref{algo:dominance_test}, is done by the complete pairwise comparison between each local nadir point in $\mathcal{N}(\mathcal{U})$ and every segment $\mathcal{L}$.

\begin{algorithm}
\scriptsize
\KwIn{$\mathcal{L}$ the local lower bound set, $\mathcal{U}$ the global upper bound set.}
\KwOut{\True \, if $\mathcal{L}$ is dominated by $\mathcal{U}$, \False \, otherwise.}

\For{$i \in [1, | \mathcal{L}| ]$}{
    $y \leftarrow\mathcal{L}[i]$ \; 
    $\lambda \leftarrow$ \text{orthogonal vector to the i-th segment of $\mathcal{L}$} 
    
    \For{$u \in \mathcal{N}(\mathcal{U})$}{
    \If{$\lambda^T u > \lambda^T y$}{\Return \False \tcp*{since the local nadir point is above $\mathcal{L}$} }
    }
}
\Return \True \;
\caption{\textsc{Fully explicit dominance test}($\mathcal{L}$, $\mathcal{U}$)} 
\label{algo:dominance_test}
\end{algorithm}

\subsubsection{Branching} In this section, we describe the branching strategies in our BOBLB\&B algorithm.

\paragraph{Branching on variables} The BOLP relaxation problem gives the lower bound set consisting of several fractional solutions, which makes it hard to choose which one of the variables to split. In our implementation, the next free variable to branch on is the next index one.


\paragraph{Branching on criteria space (Pareto branching)} The motivation for the objective branching in the bi-objective context is to partition the searching area from the criteria space. \cite{stidsen2018hybrid} suggested the slicing operation that divides the criteria space into equilibrated cones with additional constraints. \cite{stidsen2014branch} proposed the binary Pareto branching that splits the non-dominated area by bounding on the incumbent set. With the hope that the more affined non-dominated area splits, the more searching space could be discarded. \cite{gadegaard2019bi} investigated the extended Pareto branching rule partitioning the non-dominated area into distinct parts as many as the number of non-dominated local nadir points.

\paragraph{Branching strategy} In our implementation \autoref{algo:pareto_branching}, we apply Pareto branching on $\mathcal{N}(\mathcal{U})$ with priority, if they are not branched in any predecessors before. Otherwise, we continue to branch on free variables. 
    
\begin{algorithm}
\scriptsize
\KwIn{$\eta$ a node in the B\&B tree. 
}

$EPB \leftarrow \True$ \tcp*{will stay \True , if the extended Pareto branching is applied}

\For{$u \in \eta.local\_nadir\_points$}{
    \If{$u \in predecessor(\eta).local\_nadir\_points$}{
    $EPB \leftarrow \False$ \;
    \text{\textbf{goto} line 11} \; 
    }
}

\If{$EPB$}{
\tcp{split node $\eta$ by extended Pareto branching}
\For{$u \in \eta.local\_nadir\_points$}{
    $\eta' \leftarrow \eta \cup \{z_1(x) \leqslant u[1], \, z_2(x) \leqslant u[2]\}$ \; 
    $Tree.add\_sucessor(\eta, \eta')$ \;
}

}\Else{ \tcp{otherwise branch on free variable}
$x' \leftarrow choose\_a\_free\_var(\eta)$ \;
$\eta'_0 \leftarrow \eta \cup \{x' = 0\}$ \; 
$\eta'_1 \leftarrow \eta \cup \{x' = 1\}$ \; 
$Tree.add\_sucessor(\eta, \eta'_0)$ \;
$Tree.add\_sucessor(\eta, \eta'_1)$ \;

}

\caption{\textsc{Branching}($\eta$)} 
\label{algo:pareto_branching}
\end{algorithm}

\paragraph{Choice of the active node} In the literature, most of the BOB\&B algorithms use the depth-first search order. Since the lower bound sets are not always comparable, the classical best-first strategy is not straightforwardly applicable. We implement three possible strategies that are the depth-first search, the breadth-first search, and the arbitrary order.

\subsubsection{BOBLB\&B algorithm}

As every component of the BOBLB\&B algorithm is detailed previously, we conclude in the following the procedure applied iteratively for each subproblem and a pseudocode description in \autoref{algo:B&B}, where every non visited node is evaluated by the lower and upper bound sets. Given a B\&B node $\eta$, solving the BOLP relaxation of BOBLP by the dichotomy method yields the exact non-dominated boundaries $\widetilde{\mathcal{Y}}(\eta)_N$ in line 7 as a valid lower bound set. If the current subproblem $BOLP(\eta)$ is unfeasible, $\eta$ is pruned by infeasibility in lines 8-9. The global upper bound set is updated with new feasible points found in lines 11-13. If the lower bound set consists of a unique integer point, $\eta$ is fathomed by integrity in lines 14-15. In lines 16-17 $\eta$ is pruned if dominated. If the current node is not fathomed, the branching rules in \autoref{algo:pareto_branching} are performed. To optimize memory space, if all successors of a node are visited and evaluated then that node is removed from the B\&B tree. The above procedure is repeated until all subproblems are evaluated.

\begin{algorithm}
\scriptsize
\KwIn{A bi-objective binary linear program BOBLP $\{\min z(x) = (z_1(x), \, z_2(x)) \text{ s.t. } x \in \mathcal{X} = \{x \in \{0,1\}^n : Ax \leqslant b\}  \}$.}
\KwOut{The set of non-dominated points $\mathcal{Y}_N$.}

$\mathcal{U} \leftarrow \emptyset $ \tcp*{ the global upper bound set}

$\eta \leftarrow \{\min z(x) = (z_1(x), \, z_2(x)) \text{ s.t. } x \in \widetilde{\mathcal{X}} = \{x \in [0,1]^n : Ax \leqslant b\}  \} $ \tcp*{initialize the root node}

$Tree \leftarrow \{\eta \}$ \;
$\mathcal{S} \leftarrow \emptyset$ \tcp*{a dictionary associating the list of equivalent solutions with a criteria point}

\While{$has\_non\_visited\_node(Tree)$}{

$\eta \leftarrow choose\_a\_non\_visited\_node(Tree) $ \;

$\mathcal{L} \leftarrow BOLP\_relaxation(\eta)$ \tcp*{compute the lower bound set}

\If{$\mathcal{L} \IS \NULL $}{ 
$free(\eta)$ \tcp*{fathoming by infeasibility}
}\Else{
\For{$y \in \mathcal{L}$}{
\If{$is\_feasible(\mathcal{S}(y))$}{\nameref{algo:update_incumbent} \;}
}

\If{$|\mathcal{L}|=1 \AND is\_feasible( \mathcal{S}(\mathcal{L}[1]))$}{
    $free(\eta)$ \tcp*{fathoming by integrity}
}\ElseIf{\nameref{algo:dominance_test}}{ 
$free(\eta)$ \tcp*{fathoming by dominance}
}\Else{
\nameref{algo:pareto_branching} \;

\If{\NOT $has\_non\_visited\_successor(\eta.predecessor)$}{$free(\eta.predecessor)$ \;}
}
}
}
\Return $\mathcal{U}$ \;
\caption{\textsc{Bi-objective Binary Linear Branch\&Bound}($BOBLP$)} 
\label{algo:B&B}
\end{algorithm}

\begin{lemma}
At the end of BOBLB\&B algorithm, all integer equivalent solutions of $\mathcal{S}(\mathcal{U})$ is $\mathcal{X}_E$ the complete efficient set.
\end{lemma}

\begin{proof}{Proof}
Suppose there exists $x \notin \mathcal{S}(\mathcal{U}) $ which is efficient. Since the branching strategy does not eliminate any feasible solution, then $x$ is pruned from a node in the BOBLB\&B tree by one of the three fathoming rules. As $x$ is a feasible solution, then $x$ can not be pruned from an infeasible node. We claim that $x$ cannot
be pruned from a node by integrity. Indeed either $x$ is the unique integer point in this node and then $x$ is not eliminated; or there are at least two different non-dominated extreme points in this node which should not be fathomed. Finally, if $x$ belongs to a node $\eta$ fathomed by dominance, then, as $x$ is efficient, $z(x) \in \mathcal{L}(\eta) + \mathbb{R}^2_{\geqq}$ where $\mathcal{L}(\eta)$ is the valid lower bound set of node $\eta$. This contradicts this pruning rule fathoming by dominance.
\hfill $\Box$
\end{proof}

\subsection{First experimental results}


Our BOBLB\&B algorithm is implemented in \textsc{Julia} language \cite{julia} and is integrated into the \textsc{vOptSolver} project \cite{vOptSolver}, which is an open-source software devoted to solving generic multi-objective linear integer programs. Our code is available on \url{https://github.com/Yue0925/vOptGeneric.jl/tree/master/src}.

Our experimental tests are run on the \textsc{MAGI} server of University Sorbonne Paris Nord, using one Linux machine with \verb!Intel(R) Xeon(R) Silver 4210R! \verb!CPU!, \verb!94 GB! of \verb!RAM!. To cope with BOLP relaxation during the dichotomy method, we invoke the simplex algorithm provided by IBM ILOG \textsc{Cplex} Optimizer \cite{cplex} solver, version 22.1.0, to optimize single-objective LP problems.

We execute the branching strategy with the breadth-first search order. The free variables are branched in the default order. For all instances tested, the Time Limit (TL) is 3600 CPU seconds.

\subsubsection{Instances}

In our computational tests, we run our algorithms on 2227 instances, with 2161 instances from the literature and 66 randomly generated instances.

\begin{itemize}
    \item[$\bullet$] Bi-Objective Bi-dimensional Knapsack Problem (BOBKP) instances from the \textsc{VoptLib} library \cite{vOptLib}: 35 instances.

    \item[$\bullet$] Bi-Objective Multi-Demand Multi-dimensional Knapsack Problem (BOMDMKP) from \cite{cappanera2005local}: 78 instances. The instances are originally single-objective, the second objective coefficients are generated with the linear correlation $\rho (z_1, z_2) \approx -0.92$.

    \item[$\bullet$] Multi-Objective Knapsack Problem (MOKP-A) from \cite{gadegaard2020branch}: 1200 instances.

    \item[$\bullet$] Multi-Objective Knapsack Problem (MOKP-B) from \cite{kirlik2014new}: 160 instances.

    \item[$\bullet$] Bi-Objective Set Covering Problem (BOSCP): 66 single-objective instances are randomly generated with the number of objects $n \in [20, 100]$, the number of covers is decided uniformly from $[\frac{n}{10}, \frac{3n}{10}]$ while the density of cover is also uniformly defined from $10\%$ to $30\%$. The second objective function is generated with a negative correlation with respect to the first objective coefficients $\rho (z_1, z_2) \approx -0.92$.

    \item[$\bullet$] Bi-Objective Set PArtitioning problem (BOSPA) collected in the ORlibrary \cite{beasley1990or}: 28 instances with second objective extended by \cite{gandibleux2021primal}.

    \item[$\bullet$] Multi-Objective Assignment Problem (MOAP) from  \cite{przybylski2010two,kirlik2014new}: 40 instances. 

    \item[$\bullet$] Multi-Objective Uncapacitated Facility Location Problem (MOUFLP) from \cite{gadegaard2020branch}: 620 instances.

\end{itemize}

The above Multi-Objective (MO) instances are considered as bi-objective ones by taking into account the two first objective functions.

\subsubsection{Comparison with state-of-the-art method}\label{sec:exper_bb}

We compare our BOBLB\&B algorithm with or without the extended Pareto branching (EPB), respectively denoted as EPB B\&B and B\&B to the state-of-the-art $\epsilon$-constraint method. In \autoref{tab:BOBKP_EPSILONvsBBvsEPBBB} for all categories instances, the columns indicate

\begin{itemize}
    \item[-] $n$ : the average number of variables 

    \item[-] $m$ : the average number of constraints

    \item[-] $|\mathcal{Y}_N|$ : the number of non-dominated points 

    \item[-] Time(s) : the average computation time in CPU seconds executed by algorithms

    \item[-] \#TL : the number of instances reaching the time limit among the total number of instances

    \item[-] \#Nodes : the average number of nodes explored during the algorithm 

\end{itemize}


\begin{table}
\centering
\scriptsize
\begin{tabular}{lrrrrrrrrrr}

\toprule
~& ~ & ~&~ & \textbf{\tiny{$\epsilon$-constraint}} & \multicolumn{3}{c}{\textbf{B\&B}}  & \multicolumn{3}{c}{\textbf{EPB B\&B}} \\

\cmidrule(r){5-5} \cmidrule(r){6-8} \cmidrule(r){9-11} 

\textbf{Instance}  & \textbf{n} & \textbf{m} &\textbf{$|\mathcal{Y}_N|$}  &  \textbf{Time(s)} & \textbf{Time(s)} & \textbf{\#TL} &\textbf{\#Nodes} & \textbf{Time(s)} & \textbf{\#TL}  &\textbf{\#Nodes} \\
\midrule

\multirow{3}{*}{BOBKP} & 10-20 & 2.00 &8.43 & 0.99 & \textbf{0.30} & 0/7 & 509.72 & 0.96 & 0/7 & 740.29 \\ 
~ & 30-40 & 2.00 &27.72 & \textbf{5.37} & 375.42 & 0/7 & 109992.57 & 409.18 & 0/7 & 256052.85 \\ 
~ & 50 & 2.00 &51.57 & \textbf{10.74} & 1192.92 & 2/7 & 165915.00 & 1359.82 & 2/7 & 386796.86 \\ 
\hline

\multirow{2}{*}{BOMDMKP} & 10-20 & 2.86 & 9.81 &  \textbf{4.88} & 12.71 & 0/21 & 15700.91 & 13.58 & 0/21 & 16071.00 \\
~ & 30-40  &4.84 & 31.47 & \textbf{13.74} & 2503.30 & 32/57 & 648192.34 & 2598.29 & 34/57 & 822636.19 \\
\hline

\multirow{3}{*}{MOKP-A} &  10-20& 1.00&10.64& 4.18& \textbf{2.90} & 0/1080 & 1506.12& 4.77 & 0/1080 &3738.68  \\
~ & 25-30 & 1.00 &  41.64& \textbf{7.27}& 19.21 & 0/60 & 8402.90 & 35.71 & 0/60  & 42977.72 \\
~ & 35-40 & 1.00 &82.85 &\textbf{12.63} & 96.33 & 0/60  & 27932.67 & 200.66 & 0/60  & 210114.13  \\
\hline

\multirow{5}{*}{MOKP-B} & 10-20 & 1.00 & 7.80& \textbf{3.84}& 6.91& 0/60 & 6516.04 & 7.59 &  0/60  & 6733.68 \\
~ & 30-40& 1.00& 26.78 & \textbf{5.91} & 354.98& 0/40 & 181982.95 & 496.96 & 0/40 &490901.95  \\
~ &50-60& 1.00 & 56.75 & \textbf{10.01} &  2772.78& 9/20 & 547773.50 & 3111.45 & 15/20 & 1.93375e6\\
~ & 70-80 &1.00& 100.60 & \textbf{16.93}& 3600.98 & 20/20 & 241904.20 &3601.47 & 20/20 & 1.71603e6  \\
~ & 90-100 & 1.00& 140.10 & \textbf{24.21} & 3600.95 & 20/20 & 143654.90 & 3600.93 & 20/20 & 1.534965e6 \\
\hline

\multirow{3}{*}{BOBSCP} & 20-40 & 6.00 &8.59 & 4.01 & \textbf{2.10} & 0/24 & 374.42 & 3.15 & 0/24 & 1928.34 \\ 
~ & 60-80 & 13.44 & 20.00 & \textbf{6.92} & 16.26 & 0/27 & 3144.93 & 31.03 & 0/27 &17551.59 \\
~ & 100 & 22.00 & 27.13 &\textbf{9.74} & 114.47 & 0/15 & 9436.73 & 141.45 & 0/15 & 62818.93 \\
\hline

\multirow{3}{*}{BOSPA} & 50-500 & 21.17 & 9.00 &  \textbf{4.91} & 15.50 & 0/6 & 3114.33 & 78.43 & 0/6 & 8317.17  \\
~ & 550-1500 & 21.94 & 14.31 & \textbf{6.64} & 249.19 & 0/16 & 16997.62 & 2187.54 & 8/16 & 75361.88 \\
~ & 1600-3000 & 25.50&13.33 & \textbf{8.41} & 1505.24 & 1/6 & 28945.67 &3601.42 & 6/6 & 58056.83  \\
\hline

\multirow{2}{*}{MOAP} & 100-225 & 25.00 & 26.75 & \textbf{6.66} & 60.94 & 0/20 &12436.50 & 253.89 & 0/20 & 49937.15  \\ 
~ & 400-625 & 45.00 & 61.15 & \textbf{11.72} & 2262.61 & 9/20 & 104902.20 & 3350.48 & 16/20 &297145.00 \\
\hline

\multirow{4}{*}{MOUFLP} & 30-42 & 41.08 & 26.70 & 6.63 & \textbf{4.11} & 0/130 & 2189.31 & 13.74 & 0/130 & 10345.31 \\
~ & 56-72& 57.23 & 43.77 & \textbf{9.68} & 11.48 & 0/130 & 5169.66 &  50.77 & 0/130 & 33243.05   \\
~ & 110 - 156& 133.00 & 10.22 & \textbf{4.52} & 9.61 & 0/240 & 2163.35 &15.05 & 0/240 &5224.98  \\
~ & 210& 210.00 & 18.14 & \textbf{7.03}&  50.17& 0/120 & 8181.37& 111.28& 0/120 & 27683.61\\

\bottomrule
\end{tabular}
\caption{A comparison of the performances of $\epsilon$-constraint, BOBLB\&B without and with EPB algorithm.}
\label{tab:BOBKP_EPSILONvsBBvsEPBBB}
\end{table}




\autoref{tab:BOBKP_EPSILONvsBBvsEPBBB} indicates the best running time in bold among $\epsilon$-constraint, B\&B and EPB B\&B methods.  It appears that $\epsilon$-constraint outperforms the other methods for the larger instances, while, for small ones, BOBLB\&B tree completes the enumeration process much faster. Indeed, the $\epsilon$-constraint method repeatedly solves a single-objective ILP in $\mathcal{O}(|\mathcal{Y}_N|)$ time, which is quite time consuming in general, but not for these above rather small instances solved efficiently by ILP solver. On the opposite, BOBLB\&B methods solve the BOLP relaxation problem within a branching tree, in which the tree size grows exponentially with the number of variables linearly increased. The ambition of this article, is to sufficiently cut down the branching tree to reduce the number of nodes and the total time.

The advantage of the $\epsilon$-constraint method relies on $\mathcal{O}(|\mathcal{Y}_N|)$ resolutions using an ILP solver with whose efficiency comes from cutting, pre-processing, reductions, heuristics, and other advanced techniques embedded inside. We aim to devise such advanced ILP techniques inside the BOBLB\&B methods to directly benefit for the unique branching tree. Note that, for instances where an optimal single-objective ILP solution is hard to achieve, the $\epsilon$-constraint method will solve a large number of slow single-objective ILPs which will impact its global performance.

During the BOBLB\&B tree exploration, we find that computing exhaustively the BOLP relaxation occupies more than 90\% of the total time consumed, as all extreme non-dominated points of BOLP are exhaustively enumerated. Therefore, we propose an alternative relaxation to speed up the whole process by searching a tradeoff between the LBS exactness and the saved computation time.


In the literature, EPB has been shown to be efficient, which is not the case here. Indeed, the global UBS is initialized with an empty set and is updated only from BOLP relaxation, the EPB starts to branch on feasible points far away from LBS. Consequently, EPB in the BOBLB\&B algorithm leads to more nodes and slows down the enumeration process. We propose to both reinforce the UBS and the LBS benefiting from the ILP techniques.




\section{Cutting methods for bi-objective integer programming}\label{sec:BOBC}

As shown numerically in the last section, the BOBLB\&B algorithm remains far from being competitive with the state-of-the-art $\epsilon$-constraint method. Analogously to the single-objective B\&B algorithm, the good quality of the bound sets is crucial to prune more redundant nodes. To reinforce the relaxation bound set, we present in this section two new generic cutting plane frameworks. The first idea is simultaneously cut several LBS points. The second idea is to use the cuts generators of ILP solvers.

\subsection{Polyhedral backgrounds}

The separation problem of single-objective combinatorial optimization is to find a valid inequality separating a fractional point from the convex hull of the feasible set. In bi-objective context, we would like to cut off several points of a given lower bound set $\mathcal{L}$. A tight LBS is the inequalities defining the dominant of the feasible points, i.e. $Dom(\mathcal{Y})=(\mathcal{Y}+\mathbb{R}^2_{\geqq})_N$. Since obtaining these inequalities  (which give $conv(\mathcal{Y})_N$) is hard, we will approach this convex hull by computing valid inequalities defined as follows.

An inequality in solution space $\alpha x\leqslant \beta$, with $\alpha \in \mathbb{R}^n$, $\beta \in \mathbb{R}$, is \textit{valid} for $conv(\mathcal{X})$ if $\forall x\in\mathcal{X}$, $\alpha x\leqslant \beta$. Analogously, an inequality in criteria space $\pi y \leqslant \omega$ with $\pi \in \mathbb{R}^2 , \, \omega\in \mathbb{R}$ is \textit{valid} for $conv(\mathcal{Y})$, if $\forall x\in \mathcal{X}$, $\pi z(x)\leqslant\omega$. Note that this is equivalent to say that $\forall y\in \mathcal{Y}$, $\pi y \leqslant \omega$.

Let us first remark that adding valid inequalities for either  $conv(\mathcal{X})$ or $conv(\mathcal{Y})$ corresponds to add valid inequalities for $conv(\mathcal{Y})$ in criteria space.

\begin{lemma}\label{lem:proj_valid_inq}
The linear projection $z(\cdot)$ of a valid inequality $\alpha x \leqslant \beta $ for $conv(\mathcal{X})$ is a polyhedron $\mathcal{P}$ in the criteria space. Moreover the inequalities defining $Dom(\mathcal{P}$) are valid for $conv(\mathcal{Y})$.
\end{lemma}

\begin{proof}{Proof}
In solution space, a valid inequality $\alpha x \leqslant \beta$ defines the halfspace $\mathcal{H} = \{x \in \mathbb{R}^n | \alpha x \leqslant \beta \}$. The linear projection $z(\cdot)$ of $\mathcal{H}$ in criteria space is a polyhedron $\mathcal{P} = z(\mathcal{H})$. Since the linear projection $z(\cdot)$ of the convex polytope $conv(\mathcal{X})$ is still a convex polytope $conv(\mathcal{Y})$ and  $\alpha x \leqslant \beta$ is valid for $conv(\mathcal{X})$, then $ conv(\mathcal{Y}) \subset \mathcal{P}+\mathbb{R}^2_{\geqq} $. Moreover the dominant boundaries $Dom(\mathcal{\mathcal{P}})$ are valid for $conv(\mathcal{Y})$.  \hfill$\Box$
\end{proof}

\begin{figure}
\begin{subfigure}[b]{0.45\textwidth}
\centering
\begin{tikzpicture}[scale=.4, domain=1:10,
bluenode/.style={shape=circle, fill=blue!70, draw=blue!70, minimum size=0.1pt, scale = 0.3},
blacknode/.style={shape=circle, fill=black!40, draw=black!40, minimum size=0.1pt, scale=0.3},
greennode/.style={shape=circle, fill=black!50!green, draw=black!50!green, minimum size=0.1pt, scale=0.3},
rednode/.style={shape=circle, fill=red!50, draw=red!50, minimum size=0.1pt, scale=0.3}
]
\def \xmin{1} \def \ymin{1} \def \xmax{10} \def \ymax{10}
\filldraw [draw=black!50!green, dashed, fill=black!30!green!30, ] (2.2, 4.5) -- (8.5, 1.2) -- (10, 1.2) -- (10, 2.5) -- (8.3, 4) -- cycle;
\node[black!50!green, scale=.8] at (5.5, 3.8) {$\mathcal{P}$};
\draw[black!50!green, thick] (2.2, 4.5) -- (8.5, 1.2) -- (10, 1.2);
\node[black!50!green, scale=.8] at (8.5, 0.5) {$Dom(\mathcal{P})$};

\draw[draw=red!50] (1.5, \ymax)--(1.5,7.5)--(2.2,4.5)--(5,2.5)--(9,1.5)--(\xmax, 1.5) ;
\filldraw [draw=blue!50, fill=blue!30] (3,8) -- (6,8) -- (8,7) -- (9,5) -- (8,3)--(4,5) -- (3,6) -- cycle;

\draw[->] (\xmin,\ymin) -- (\xmax,\ymin) node[right] {$z_1$}; 
\draw[->] (\xmin,\ymin) -- (\xmin, \ymax) node[left] {$z_2$};

\node[rednode] at (1.5,7.5) () {};
\node[rednode] at (1.85,6) () {};
\node[rednode] at (2.2, 4.5) () {};
\node[greennode] at (5, 2.5) () {};
\node[black!50!green, scale=.8] at (5, 2) {$\tilde{y}$};
\node[rednode] at (8,1.75) () {};
\node[rednode] at (9,1.5) () {};
\node[red!50, scale=.8] at (3,3) {$\mathcal{\widetilde{Y}}_N$};

\node[blue!50, scale=.8] at (8, 8) {$conv (\mathcal{Y})$};
\node[blacknode] at (3,6) () {};
\node[blacknode] at (4,5) () {};
\node[blacknode] at (3,8) () {};
\node[blacknode] at (5,7) () {};
\node[blacknode] at (6,5) () {};
\node[blacknode] at (6,8) () {};
\node[blacknode] at (9,5) () {};
\node[blacknode] at (7,6) () {};
\node[blacknode] at (8,3) () {};
\node[blacknode] at (8,7) () {};

\end{tikzpicture}

  \caption{ }
\end{subfigure}
 \hfill
\begin{subfigure}[b]{0.45\textwidth}
\centering
\begin{tikzpicture}[scale=.4, domain=1:10,
bluenode/.style={shape=circle, fill=blue!70, draw=blue!70, minimum size=0.1pt, scale = 0.3},
blacknode/.style={shape=circle, fill=black!40, draw=black!40, minimum size=0.1pt, scale=0.3},
greennode/.style={shape=circle, fill=black!50!green, draw=black!50!green, minimum size=0.1pt, scale=0.3},
rednode/.style={shape=circle, fill=red!50, draw=red!50, minimum size=0.1pt, scale=0.3}
]
\def \xmin{1} \def \ymin{1} \def \xmax{10} \def \ymax{10}

\draw[draw=red!50] (1.5, \ymax)--(1.5,7.5)--(2.2,4.5)--(7,2)--(9,1.5)--(\xmax, 1.5) ;
\filldraw [draw=blue!50, fill=blue!30] (3,8) -- (6,8) -- (8,7) -- (9,5) -- (8,3)--(4,5) -- (3,6) -- cycle;

\draw[->] (\xmin,\ymin) -- (\xmax,\ymin) node[right] {$z_1$}; 
\draw[->] (\xmin,\ymin) -- (\xmin, \ymax) node[left] {$z_2$};

\node[rednode] at (1.5,7.5) () {};
\node[rednode] at (1.85,6) () {};
\node[rednode] at (2.2, 4.5) () {};
\node[greennode] at (5, 2.5) () {};
\node[black!50!green, scale=.8] at (5, 2) {$\tilde{y}$};
\node[rednode] at (7,2) () {};
\node[rednode] at (8,1.75) () {};
\node[rednode] at (9,1.5) () {};
\node[red!50, scale=.8] at (3,3) {$\mathcal{\widetilde{Y}}_N$};

\node[blue!50, scale=.8] at (8, 8) {$conv (\mathcal{Y})$};
\node[blacknode] at (3,6) () {};
\node[blacknode] at (4,5) () {};
\node[blacknode] at (3,8) () {};
\node[blacknode] at (5,7) () {};
\node[blacknode] at (6,5) () {};
\node[blacknode] at (6,8) () {};
\node[blacknode] at (9,5) () {};
\node[blacknode] at (7,6) () {};
\node[blacknode] at (8,3) () {};
\node[blacknode] at (8,7) () {};

\end{tikzpicture}
  \caption{ }
\end{subfigure}
\caption{New LBS of the intersection between $\widetilde{\mathcal{Y}}_N$ and $Dom(\mathcal{P})$
}
\label{fig:image_valid_inq}
\end{figure}
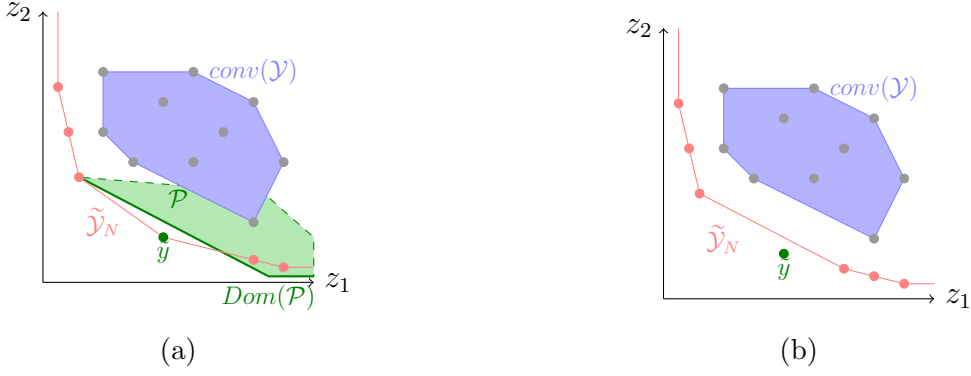

Our purpose is to generate a set of valid inequalities cutting the lower bound set $\mathcal{L}$. By \autoref{lem:proj_valid_inq}, the convex intersection of a LBS $\mathcal{L}$ and $Dom(\mathcal{P})$ forms a tighter 
lower bound set than $\widetilde{\mathcal{Y}}_N$. \autoref{fig:image_valid_inq} illustrates in (a) a polyhedron $\mathcal{P}$ in the criteria space defined as the image of a valid inequality in the solution space and (b) the resulting LBS from the intersection between the former LBS and the boundaries of $Dom(\mathcal{P})$.

\begin{problem}\label{prob:1}(\textbf{Solution space separation problem})
Given a solution $\tilde{x} \in \widetilde{\mathcal{X}}$, return $(\True, \,  \alpha x\leqslant \beta)$ if there exists a valid inequality separating $\tilde{x}$ from $conv(\mathcal{X})$ i.e. $\alpha \tilde{x} > \beta$. Otherwise the output is $(\False, \, \emptyset)$.
\end{problem}

\begin{problem}\label{prob:2}(\textbf{Criteria space separation problem})
Given a point $\tilde{y} \in \widetilde{\mathcal{Y}}$, return $(\True, \,  \pi y \leqslant \omega)$ if there exists a valid inequality separating $\tilde{y}$ from $conv(\mathcal{Y})$ i.e. $\pi \tilde{y} > \omega$. Otherwise the output is $(\False, \, \emptyset)$.
\end{problem}

The difficulty of \autoref{prob:2} in searching for a valid inequality directly from the criteria space is to prove the validity of the inverse image $z^{-1}(\pi y \leqslant \omega)$ in solution space. \cite{bokler2024outer} proposed a single objective Mixed Integer Linear Problem (MILP) to search a valid inequality $\pi y \leqslant \omega$ by verifying the whole feasible set $\mathcal{X}$. A tight LBS is obtained by a cutting-plane algorithm repeatedly solving such a MILP problem. Unfortunately, this approach is highly time consuming.

We propose another separation problem which combines the two previous separation problems.

\begin{problem}\label{prob:3}(\textbf{Multi-objective separation problem}) Let us consider a relaxation set $\mathcal{X}'$ such that $conv(\mathcal{X}) \subseteq \mathcal{X}' \subseteq \widetilde{\mathcal{X}}$.
Given ($\lambda, \tilde{x}, \, \tilde{y}=z(\tilde{x})$) where $\lambda \in \mathbb{R}^2_{\geqq}$, $\tilde{x} \in \argmin_{x \in \mathcal{X}'} \lambda^Tz(x)$, return $\True$ if there exists a valid inequality violating $\tilde{x}$ in solution space, or there exists a valid inequality violating $\tilde{y}$ in criteria space. Otherwise the output is $\False$.
\end{problem}

To answer \autoref{prob:3}, the following lemma proposes to generate inequalities in the criteria space.

\begin{lemma}\label{lem:lambda_valid_cut}
Let $\lambda \in \mathbb{R}^2_{\geqq}$ be a positive vector and $\tilde{x} \in \argmin_{x \in {\mathcal{X}'}} \lambda^T z(x)$. Consider a polyhedron $\widehat{\mathcal{X}}$ obtained from $\mathcal{X}'$ by adding inequalities of an oracle separating $\tilde{x}$ and $\mathcal{X}'$, and $\hat{x} \in \argmin_{x \in {\widehat{\mathcal{X}}}} \lambda^T z(x)$. The inequality 
\begin{equation}\label{eq:1}
\lambda^\perp y \geqslant \lambda^\perp z(\hat{x}) 
\end{equation}
is a valid inequality in criteria space separating $\tilde{y} = z(\tilde{x})$ from $conv(\mathcal{Y})$.
\end{lemma}

\begin{proof}{Proof}
Denote $\mathcal{H}_\lambda = \{y \in \mathbb{R}^2 |\lambda^\perp  y = \lambda^\perp z(\hat{x}) \}$ the hyperplane passes through the point $\hat{y} = z(\hat{x})$ and is perpendicular with the normal $\lambda$. 
Since $conv(\mathcal{X}) \subseteq \widehat{\mathcal{X}}$, $\lambda^T y \geqslant \lambda^T \hat{y}$, $\forall y \in \mathcal{Y}$, hence $ \lambda^\perp y \geqslant \lambda^\perp z(\hat{x})  $ is valid for $conv(\mathcal{Y})$. $\tilde{y}$ is violated by the hyperplane $\mathcal{H}_\lambda$, as $\tilde{x} \notin \widehat{\mathcal{X}}$, $\lambda^T \tilde{y} < \lambda^T \hat{y}$.
 \hfill$\Box$
\end{proof}

Note that, in the assumption of \autoref{lem:lambda_valid_cut}, the additional inequalities of $\widehat{\mathcal{X}}$ can either be added to a relaxation of $conv(\mathcal{X})$ or not. Indeed, they are obtained from an oracle which can be an exterior ILP solver, and it is hard to retrieve them explicitly in practice.

In the rest of this section we propose, in \autoref{sec:multipoint_cuts}, to simultaneously separate, in solution space, the numerous LBS points, and in \autoref{sec:solver_cuts}, to use \autoref{lem:lambda_valid_cut} to generate valid inequalities for the criteria space with the help of the strength of ILP solvers. Finally, in \autoref{sec:cuts_aggregation}, we discuss the way to achieve a better balance between the exactness of LBS and the computation time.

\subsection{Multi-point cutting plane}\label{sec:multipoint_cuts}

Considering the large BOLP non-dominated boundaries $\widetilde{\mathcal{Y}}_N$, rather than cutting individually every extreme point in $\widetilde{\mathcal{Y}}_N$, we would like to generate an efficient valid inequality in solution space that separates multiple lower bounds in the criteria space simultaneously. In this section we design a generic multi-point cutting plane scheme for bi-objective problems.

\subsubsection{Multi-point separation}

We define below the multi-point separation problem for bi-objective context.

\begin{problem}\label{def:multi_point_sep}(\textbf{Multi-objective multi-point separation problem})
Consider a lower bound set $\mathcal{L}$, a dictionary $\mathcal{S}$ associating all equivalent solutions imaging the same point in the criteria space, and $\mathcal{L}' \subseteq \mathcal{L}$ a subset of the lower bound set. Return $\True$ if there exists a valid inequality in criteria space separating all points in $\mathcal{L}'$ from $\mathcal{Y}$, or equivalently if there exists a valid inequality in solution space separating all points in $\bigcup_{y \in \mathcal{L}'} \mathcal{S}(y)$ from $\mathcal{X}$; otherwise return $\False$.
\end{problem}

\cite{ben2006constraint} show that deciding whether there exists an inequality separating multiple points in solution space is polynomial whenever separating one point is. However it is known that separating efficient inequality classes is often $\mathcal{NP}$-complete. Separating a point $y$ in the criteria space requires separating all equivalent solutions $\mathcal{S}(y)$, which is generally difficult to obtain. Separating multiple points in criteria space is then even harder.

Considering the two points on the left and right $\tilde{y}^l$ and $\tilde{y}^r$ of a given subset $\mathcal{L}'$, we propose to heuristically separate only $\mathcal{S}(\tilde{y}^l)$ and $\mathcal{S}(\tilde{y}^r)$. Cutting off $\mathcal{S}(\tilde{y}^l)$ and $\mathcal{S}(\tilde{y}^r)$ does not necessarily separate $conv(\mathcal{L}')$, but it is possible that the image $z(\cdot)$ of such a valid inequality $\alpha x \leqslant \beta$ would separate all points under the convex combination of $\tilde{y}^l$ and $\tilde{y}^r$ as illustrated in \autoref{fig:cuts_multi_points}. The efficiency of our multi-point separation heuristic is experimentally shown in \autoref{sec:exp_multi_cuts}.

\begin{figure}
    \centering
\begin{tikzpicture}[scale=.4, domain=1:10,
bluenode/.style={shape=circle, fill=blue!70, draw=blue!70, minimum size=0.1pt, scale = 0.3},
blacknode/.style={shape=circle, fill=black!40, draw=black!40, minimum size=0.1pt, scale=0.3},
greennode/.style={shape=circle, fill=black!50!green, draw=black!50!green, minimum size=0.1pt, scale=0.3},
rednode/.style={shape=circle, fill=red!50, draw=red!50, minimum size=0.1pt, scale=0.3}
]
\def \xmin{1} \def \ymin{1} \def \xmax{10} \def \ymax{10}
\draw[draw=red!50] (1.5, \ymax)--(1.5,7.5)--(2.2,4.5)--(5,2.5)--(9,1.5)--(\xmax, 1.5) ;
\filldraw [draw=blue!50, fill=blue!30] (3,8) -- (6,8) -- (8,7) -- (9,5) -- (8,3)--(4,5) -- (3,6) -- cycle;

\draw[->] (\xmin,\ymin) -- (\xmax,\ymin) node[right] {$z_1$}; 
\draw[->] (\xmin,\ymin) -- (\xmin, \ymax) node[left] {$z_2$};

\node[rednode] at (1.5,7.5) () {};
\node[rednode] at (1.85,6) () {};
\node[rednode] at (2.2, 4.5) () {};
\node[rednode] at (5, 2.5) () {};
\node[rednode] at (8,1.75) () {};
\node[rednode] at (9,1.5) () {};
\node[red!50, scale=.8] at (3,3) {$\mathcal{\widetilde{Y}}_N$};

\node[blue!50, scale=.8] at (8, 8) {$conv (\mathcal{Y})$};
\node[blacknode] at (3,6) () {};
\node[blacknode] at (4,5) () {};
\node[blacknode] at (3,8) () {};
\node[blacknode] at (5,7) () {};
\node[blacknode] at (6,5) () {};
\node[blacknode] at (6,8) () {};
\node[blacknode] at (9,5) () {};
\node[blacknode] at (7,6) () {};
\node[blacknode] at (8,3) () {};
\node[blacknode] at (8,7) () {};

\node[red!50, scale=.8] at (1.5,6) {$\tilde{y}_1$};
\node[red!50, scale=.8] at (2.2, 4) {$\tilde{y}_2$};
\node[red!50, scale=.8] at (5, 2) {$\tilde{y}_3$};

\filldraw [draw=red!50, fill=red!30](1.85,6) -- (2.2, 4.5) -- (5, 2.5) -- cycle;
\draw[black!50!green, thick, dashed] (2, 9) -- (1.2, 7.5) -- (3.6, 4.2) -- (7.5,1.4) -- (9, 1.4);
\node[black!50!green, scale=.8] at (8.5,2.5) {$\big(z( \alpha x \leq \beta)\big)_N$ };

\end{tikzpicture}
    \caption{An example of the multi-point separation in bi-criteria space.}
    \label{fig:cuts_multi_points}
\end{figure}
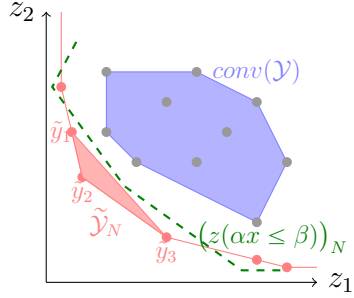

\subsubsection{Multi-point cutting plane algorithm}

Our heuristic algorithm is presented in \autoref{algo:MP_cutting_plane}. Lines 2-14 go through $\mathcal{L}$ of the current node $\eta$ in a natural order, and group the $\Delta = |\mathcal{L}'|$ consecutive points passed in separators to see if there exists a multi-point valid inequality; if not, algorithm continues to try on a smaller group of points until $\Delta = 1$ where it arrives at the classic single-point separation problem. For an appropriate tradeoff between the tight LBS and the computation time on the reoptimization in line 3, the above procedure is repeated until the termination condition is reached, that is either at an iteration if less than 60\% of points in LBS are cut off, or the cutting planes loop (line 2) reaches the maximum limit of $5$ iterations. To avoid generating redundant cuts, previously generated inequalities are stored in a cut pool implemented with a hashmap structure for each node. More precisely, the predecessor's cut pool is verified before the cutting plane loop.

\begin{algorithm}
\scriptsize
\KwIn{$\mathcal{L}$ the lower bound set (ordered from the left-up to the right-bottom in bi-criteria space) of the current node $\eta$.}
\KwOut{The new lower bound set $\mathcal{L}$ after applying or not valid inequalities.}

$\mathcal{I} \leftarrow \emptyset$ \;

\While{$end\_condition(\mathcal{L})$}{
$\mathcal{L} \leftarrow BOLP\_relaxation(\eta)$ \; 
$l \leftarrow 1$ \; 

\While{$l \leqslant |\mathcal{L}|$}{
\For{$\Delta = max\_step \ldots 1$}{
    $r \leftarrow l + \Delta$ -1\; 
    \If{$r < |\mathcal{L}|$}{
    
    $(\alpha, \beta) \leftarrow multi\_point\_separator(\mathcal{S}(\mathcal{L}[l]), \mathcal{S}(\mathcal{L}[r]))$ \; 
    
    \If{$(\alpha, \beta)$ \NOT (\NULL, \NULL) }{$\mathcal{I} \leftarrow \mathcal{I} \cup \{\alpha x \leqslant \beta \}$ \; 
    $l \leftarrow r$ \;
    \text{\textbf{goto} line 14} \;
    }
    }
} 

$l \leftarrow l +1$ \;
} 

} 

$update\_cuts\_pool(\eta, \mathcal{I})$ \; 
\Return $\mathcal{L}$ \; 
\caption{\textsc{Multi-point cutting plane}($\eta$, $\mathcal{L}$)} 
\label{algo:MP_cutting_plane}
\end{algorithm}

To see the efficiency of our multi-point cutting plane algorithm, cover inequalities, whose efficiency is well known for Knapsack problems classes, are applied on BOKP, BOMDMKP, MOKP-A and MOKP-B instances. The classical greedy generator for the cover inequalities is implemented and adapted to separate one or a subset of fractional points. Note that this technique can be extended to any heuristic or exact generators of other inequality classes.

\subsection{Cutting approach invoking ILP solver}\label{sec:solver_cuts}

In this section, a framework to provide efficient valid inequalities (\ref{eq:1}) from \autoref{lem:lambda_valid_cut} is proposed, with the algorithm to construct the new associated LBS.

\subsubsection{Cutting separation}

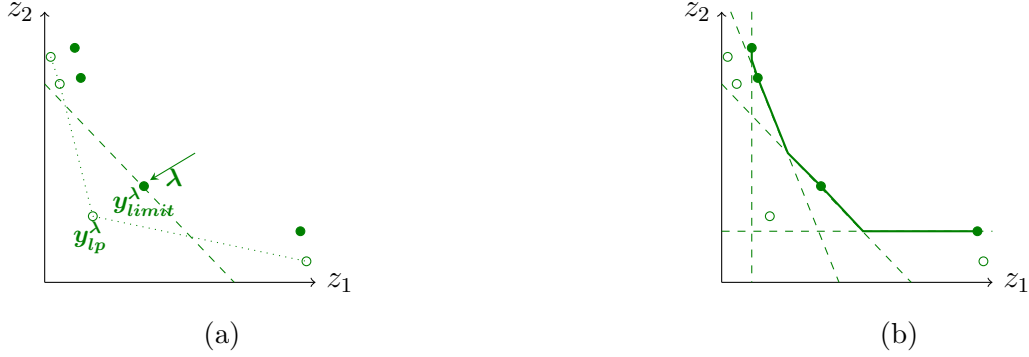
\begin{figure}
\centering
\begin{subfigure}[b]{0.4\textwidth}%
\begin{tikzpicture}[scale=.4,  domain=0:9,
bluenode/.style={shape=circle, fill=blue!70, draw=blue!70, minimum size=0.1pt, scale = 0.3},
blacknode/.style={shape=circle, fill=black!40, draw=black!40, minimum size=0.1pt, scale=0.3},
greennode/.style={shape=circle, fill=black!50!green, draw=black!50!green, minimum size=0.1pt, scale=0.3},
rednode/.style={shape=circle, fill=red!50, draw=red!50, minimum size=0.1pt, scale=0.3},
dashednodegreen/.style={shape=circle, draw=black!50!green, minimum size=0.2pt, scale=0.3},
dashednodered/.style={shape=circle, draw=red!70, minimum size=0.2pt, scale=0.3}
]

\def \xmin{0} \def \ymin{0} \def \xmax{9} \def \ymax{9}
\draw[->] (\xmin,\ymin) -- (\xmax,\ymin) node[right] {$z_1$}; 
\draw[->] (\xmin,\ymin) -- (\xmin, \ymax) node[left] {$z_2$};
\node[greennode] at (1,7.8) () {};
\node[greennode] at (1.2, 6.8) () {};
\node[greennode] at (3.3,3.2) () {};
\node[greennode] at (8.5,1.7) () {};

\node[dashednodegreen] at (0.2, 7.5) () {};
\node[dashednodegreen] at (8.7, 0.7) () {};
\node[dashednodegreen] at (0.5, 6.6) () {};
\node[dashednodegreen] at (1.6, 2.2) () {};

\draw[black!50!green, dotted, thin] (0.2, 7.5) -- (0.5, 6.6) -- (1.6, 2.2)-- (8.7, 0.7) ;

\draw [black!50!green, -stealth] (5, 4.3)  -- (3.5,3.4);
\node[black!50!green, scale=.8] at (4.3, 3.5) {$\boldsymbol{\lambda}$};

\node[black!50!green, scale=.8] at (1.5, 1.5){$\boldsymbol{y^\lambda_{lp}}$};
\node[black!50!green, scale=.8] at (3.3, 2.7){$\boldsymbol{y^\lambda_{limit}}$};

\draw[black!50!green, dashed, thin] (0, 6.6) -- (6.3, 0);
\end{tikzpicture}
\caption{ }
\label{fig:cuts_solvers_subfig1}
\end{subfigure}
\hfill
\begin{subfigure}[b]{0.4\textwidth}
\begin{tikzpicture}[scale=.4,  domain=0:9,
bluenode/.style={shape=circle, fill=blue!70, draw=blue!70, minimum size=0.1pt, scale = 0.3},
blacknode/.style={shape=circle, fill=black!40, draw=black!40, minimum size=0.1pt, scale=0.3},
greennode/.style={shape=circle, fill=black!50!green, draw=black!50!green, minimum size=0.1pt, scale=0.3},
rednode/.style={shape=circle, fill=red!50, draw=red!50, minimum size=0.1pt, scale=0.3},
dashednodegreen/.style={shape=circle, draw=black!50!green, minimum size=0.2pt, scale=0.3},
dashednodered/.style={shape=circle, draw=red!70, minimum size=0.2pt, scale=0.3}
]

\def \xmin{0} \def \ymin{0} \def \xmax{9} \def \ymax{9}
\draw[->] (\xmin,\ymin) -- (\xmax,\ymin) node[right] {$z_1$}; 
\draw[->] (\xmin,\ymin) -- (\xmin, \ymax) node[left] {$z_2$};
\node[greennode] at (1,7.8) () {};
\node[greennode] at (1.2, 6.8) () {};
\node[greennode] at (3.3,3.2) () {};
\node[greennode] at (8.5,1.7) () {};

\node[dashednodegreen] at (0.2, 7.5) () {};
\node[dashednodegreen] at (8.7, 0.7) () {};
\node[dashednodegreen] at (0.5, 6.6) () {};
\node[dashednodegreen] at (1.6, 2.2) () {};

\draw[black!50!green, dashed, thin] (1,9) -- (1, 0);
\draw[black!50!green, dashed, thin] (0, 1.7) -- (9, 1.7);
\draw[black!50!green, dashed, thin] (0.3, 9) -- (3.9, 0);
\draw[black!50!green, dashed, thin] (0, 6.6) -- (6.3, 0);

\draw[black!50!green, thick] (1, 7.8) -- (1, 7.4) -- (1.2, 6.8) -- (2.2, 4.3) -- (3.3,3.2) -- (4.7, 1.7) -- (8.5,1.7);

\end{tikzpicture}
\caption{ }
\label{fig:cuts_solvers_subfig2}
\end{subfigure}
    
\caption{The new lower bound set invoking the ILP solver for cutting approach.}    
\label{fig:cuts_solvers}
\end{figure}

Given a parameter $\lambda \in \mathbb{R}^2_{\geqq}$ and $\sum_{k=1}^{2} \lambda_k = 1$, the scalarization single-objective linear program is $\widetilde{P}(\lambda) = \{\min \lambda^T z(x) \text{ s.t. } x \in \widetilde{\mathcal{X}} \}$. The optimal value $y^\lambda_{lp}$ of $\widetilde{P}(\lambda)$ is then a point of the lower bound set illustrated in \autoref{fig:cuts_solvers_subfig1}. Instead of solving $\widetilde{P}(\lambda)$, it is possible to let ILP solver deal with the integer problem $P(\lambda) = \{\min \lambda^T z(x) \text{ s.t. } x \in \mathcal{X} \}$ but with a limit (e.g. the CPU time or the number of explored nodes) to keep a low global consuming time. If a primal solution is obtained, denote the best primal solution encountered so far $y^\lambda_{limit} = z(x^\lambda_{limit})$, and the hyperplane $\mathcal{H}^\lambda_{\geqq} = \{ y \in \mathbb{R}^2 | \lambda^\perp y \geqslant \lambda^\perp y^\lambda_{limit} \}$ is then a valid inequality cutting $y^\lambda_{lp}$ (see \autoref{fig:cuts_solvers_subfig1}). Moreover, considering a set $\Lambda$ of scalarization vectors, the intersection of every halfspace $\bigcap_{\lambda \in \Lambda} \mathcal{H}^\lambda_{\geqq} $ is convex and the non-dominated boundaries $\big( \bigcap_{\lambda\in\Lambda} \mathcal{H}^\lambda_{\geqq} \big)_N $ (bold lines in \autoref{fig:cuts_solvers_subfig2}) form a valid LBS.

As ILP solver is invoked independently for $P(\lambda), \, \forall \lambda \in \Lambda$, the best primal points $y^\lambda_{limit}$ are not the extreme points on the same cut polyhedron. Therefore, the new LBS is not the convex combination of $\{y^\lambda_{limit} \,  \forall \lambda \in \Lambda\}$, but is the intersection of halfspaces presented in \autoref{algo:blackbox_LBS_cut} and \autoref{algo:LBS_update_intersect}.

\subsubsection{Convex intersection between lower bound sets}

\autoref{algo:blackbox_LBS_cut} invokes an ILP optimizer to solve a series of scalarization problems over an inequality set $\mathcal{X}'$ in a dichotomic order, and compute their convex intersection in \autoref{algo:LBS_update_intersect}. 

At each iteration of \autoref{algo:blackbox_LBS_cut} between lines 14-22, the optimizing parameter $\lambda$ is defined as the vector perpendicular to the segment linking the last two encountered consecutive points.

\begin{figure}
\centering
\begin{subfigure}[b]{0.4\textwidth}%
\begin{tikzpicture}[scale=.4,  domain=0:9,
bluenode/.style={shape=circle, fill=blue!70, draw=blue!70, minimum size=0.1pt, scale = 0.3},
blacknode/.style={shape=circle, fill=black!40, draw=black!40, minimum size=0.1pt, scale=0.3},
greennode/.style={shape=circle, fill=black!50!green, draw=black!50!green, minimum size=0.1pt, scale=0.3},
rednode/.style={shape=circle, fill=red!50, draw=red!50, minimum size=0.1pt, scale=0.3},
dashednodegreen/.style={shape=circle, draw=black!50!green, minimum size=0.2pt, scale=0.3},
dashednodered/.style={shape=circle, draw=red!70, minimum size=0.2pt, scale=0.3}
]

\def \xmin{0} \def \ymin{0} \def \xmax{9} \def \ymax{9}
\draw[->] (\xmin,\ymin) -- (\xmax,\ymin) node[right] {$z_1$}; 
\draw[->] (\xmin,\ymin) -- (\xmin, \ymax) node[left] {$z_2$};

\fill[fill=gray!15] (1, 8.6)-- (1.2, 6.8) -- (2.2, 4.3) -- (3.3,3.2) -- (4.7, 1.7) -- (5.5,1.7)  --cycle;

\node[greennode] at (1,7.8) () {};
\node[greennode] at (1.2, 6.8) () {};
\node[greennode] at (3.3,3.2) () {};
\node[greennode] at (8.5,1.7) () {};
\draw[black!50!green, thick] (1, 9) -- (1, 7.4) -- (1.2, 6.8) -- (2.2, 4.3) -- (3.3,3.2) -- (4.7, 1.7) -- (9,1.7);

\draw[black!50!green, dotted, semithick] (1.2, 6.8) -- (3.3,3.2);
\draw [black!50!green, -stealth] (3, 5) -- (1.6, 4);
\node[black!50!green, scale=.8] at (3.2, 4.5) {$\boldsymbol{\lambda}$};

\node[greennode] at (2.8, 5.8) () {};
\node[black!50!green, scale=.8] at (3.3, 6.3){$\boldsymbol{y_{limit}^\lambda}$};

\draw[black!50!green, dashed, thin] (0.8, 9) -- (6.4, 0);
\end{tikzpicture}
\caption{ }
\label{fig:intersection_hyoerplane1}
\end{subfigure}
\hfill
\begin{subfigure}[b]{0.4\textwidth}
\begin{tikzpicture}[scale=.4,  domain=0:9,
bluenode/.style={shape=circle, fill=blue!70, draw=blue!70, minimum size=0.1pt, scale = 0.3},
blacknode/.style={shape=circle, fill=black!40, draw=black!40, minimum size=0.1pt, scale=0.3},
greennode/.style={shape=circle, fill=black!50!green, draw=black!50!green, minimum size=0.1pt, scale=0.3},
rednode/.style={shape=circle, fill=red!50, draw=red!50, minimum size=0.1pt, scale=0.3},
dashednodegreen/.style={shape=circle, draw=black!50!green, minimum size=0.2pt, scale=0.3},
dashednodered/.style={shape=circle, draw=red!70, minimum size=0.2pt, scale=0.3}
]
\def \xmin{0} \def \ymin{0} \def \xmax{9} \def \ymax{9}
\draw[->] (\xmin,\ymin) -- (\xmax,\ymin) node[right] {$z_1$}; 
\draw[->] (\xmin,\ymin) -- (\xmin, \ymax) node[left] {$z_2$};

\fill[fill=gray!15] (4, 2.4) -- (4.7, 1.7) -- (5.8 ,1.7)  --cycle;

\node[greennode] at (1,7.8) () {};
\node[greennode] at (1.2, 6.8) () {};
\node[greennode] at (3.3,3.2) () {};
\node[greennode] at (8.5,1.7) () {};
\draw[black!50!green, thick] (1, 9) -- (1, 7.4) -- (1.2, 6.8) -- (2.2, 4.3) -- (3.3,3.2) -- (4.7, 1.7) -- (9,1.7);

\draw[black!50!green, dotted, semithick] (3.3,3.2) -- (8.5,1.7) ;
\draw [black!50!green, -stealth] (6.8, 4) -- (6, 2.5);
\node[black!50!green, scale=.8] at (7, 3) {$\boldsymbol{\lambda}$};

\node[greennode] at (3.7, 2.5) () {};
\node[black!50!green, scale=.8] at (3.4, 2.1){$\boldsymbol{y_{limit}^\lambda}$};
\draw[black!50!green, dashed, thin] (0, 3.8) -- (9, 0.6);
\end{tikzpicture}
\caption{ }
\label{fig:intersection_hyoerplane2}
\end{subfigure}

\caption{The two cases of $y^\lambda_{limit}$ induced by \autoref{algo:blackbox_LBS_cut}.}
\label{fig:intersection_hyoerplane}
\end{figure}
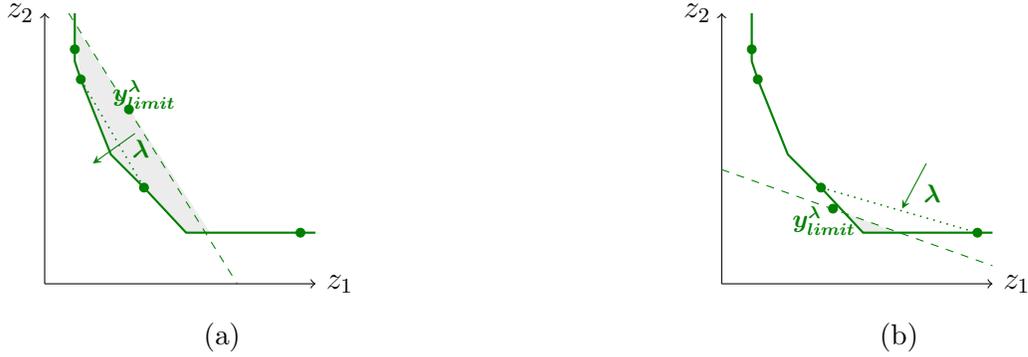

\begin{algorithm}
\scriptsize
\KwIn{A linear system $\mathcal{X}'\subseteq \widetilde{\mathcal{X}}$. $max\_time\_limit$ and $max\_node\_limit$ parameters to be set in ILP solver.
}

\KwOut{A lower bound set $\mathcal{L}$.}

$\mathcal{L} \leftarrow \{\}$ \;
$Optimizer.time\_limit \leftarrow max\_time\_limit $ \; 
$Optimizer.node\_limit \leftarrow max\_node\_limit$ \;

$\text{Set function } P(\lambda) \text{ as } \{\min \lambda^T z(x) \text{ s.t. } x \in \mathcal{X}' \}$ \;

$\lambda_l \leftarrow [1.0, 0.0]$ \tcp*{leftmost scalar}
$y^{\lambda_l} \leftarrow Optimizer.solve(P(\lambda_l))$ \; 

\textsc{Update}($\mathcal{L}$, $y^{\lambda_l}$) \; 

$\lambda_r \leftarrow [0.0, 1.0]$ \tcp*{rightmost scalar}
$y^{\lambda_r} \leftarrow Optimizer.solve(P(\lambda_r))$ \; 
\textsc{Update}($\mathcal{L}$, $y^{\lambda_r}$) \; 

\If{$|\mathcal{L}| < 2 $}{\Return $\mathcal{L}$ \; }

$pairs \leftarrow \{(\mathcal{L}.first(), \mathcal{L}.last())\}$ 

\While{$|pairs| > 0 $}{
$(y^{\lambda_l}, y^{\lambda_r}) \leftarrow pairs.first() $ \; 
$remove(pairs.first)$ \;
$\lambda \leftarrow [y^{\lambda_l}[2]- y^{\lambda_r}[2], \, y^{\lambda_r}[1]- y^{\lambda_l}[1]]$ \; 
$y^{\lambda} \leftarrow Optimizer.solve(P(\lambda))$ \;

$inserted \leftarrow$ \nameref{algo:LBS_update_intersect} \;

\If{$inserted$}{ 
Set $(y^l, y^r)$ as the left and right points adjacent to $y^{\lambda}$ in $\mathcal{L}$ \;

$pairs \leftarrow pairs \cup \{(y^l, y^\lambda) , (y^\lambda, y^r)\}$ \;
}

} 
\Return $\mathcal{L}$ \; 
\caption{\textsc{ILP cutting separation}($\mathcal{X}', \, max\_time\_limit, \, max\_node\_limit$)} 
\label{algo:blackbox_LBS_cut}
\end{algorithm}

The ILP optimizer may yield points corresponding to two cases.
One case is when $y^\lambda_{limit}$ is strictly above the segment connected by the two previous consecutive points. The other case is when $y^\lambda_{limit}$ is under the segment of the two previous points and even below the LBS. For instance in \autoref{fig:intersection_hyoerplane1} the new output $y_{limit}^\lambda$ from the ILP optimizer is strictly above the (dotted) line segment of two consecutive points, the part of $\mathcal{L}+\mathbb{R}^2_{\geqq}$ in gray below the intersecting hyperplane is eliminated. In \autoref{fig:intersection_hyoerplane2}, the new output $y_{limit}^\lambda$ is strictly under $\mathcal{L}$, the part of $\mathcal{L}+\mathbb{R}^2_{\geqq}$ in gray below the intersecting hyperplane and the point $y_{limit}^\lambda$ are both excluded.

\begin{algorithm}
\scriptsize
\KwIn{The current lower bound set $\mathcal{L}$ and the new primal point found $y^\lambda$.}
\KwOut{$\True$ if $y^\lambda$ is inserted into $\mathcal{L}$ after convexly intersecting $\mathcal{L}$ with the hyperplane passing by $y^\lambda$, $\False$ otherwise.}

\tcp{find the intersection points between the current hyperplane and $\mathcal{L}$} 
$intersection\_pts \leftarrow intersect(\mathcal{L}, y^{\lambda})$\;

\For{$y^{\lambda'} \in \mathcal{L}$}{
\tcp{eliminate all points that are below the hyperplane passing by $y^\lambda$}
\If{$\lambda^Ty^{\lambda'} < \lambda^T y^\lambda$}{ $delete(\mathcal{L}, y^{\lambda'})$ \; }
} 

\For{$y \in intersection\_pts$}{$\mathcal{L} \leftarrow \mathcal{L} \cup y$ \tcp*{add the intersecting points in $\mathcal{L}$}}

\tcp{ignore $y^\lambda$ in case of below $\mathcal{L}$}
\For{$y^{\lambda'} \in \mathcal{L}$}{
\If{$\lambda'^T y^\lambda < \lambda'^T y^{\lambda'}$}{\Return \False \;}
}

$\mathcal{L} \leftarrow \mathcal{L} \cup y^\lambda$ \;
\Return \True \; 
\caption{\textsc{Update}($\mathcal{L}$, $y^\lambda$)} 
\label{algo:LBS_update_intersect}
\end{algorithm}

To obtain the best lower bound set induced by a point $y^\lambda$, the function  \nameref{algo:LBS_update_intersect} of \autoref{algo:LBS_update_intersect} copes with the convex intersection of the ongoing LBS $\mathcal{L}$ and the hyperplane $\mathcal{H}^\lambda$ passing by $y^\lambda$. To this end, this function eliminates all points under the current intersecting hyperplane $\mathcal{H}^\lambda$ with $\mathcal{L}$, and then adds to $\mathcal{L}$ the new intersected points between $\mathcal{L}$ and $\mathcal{H}^\lambda$. If $y^\lambda$ is outside $\mathcal{L}+\mathbb{R}^2_{\geqq}$, it is dominated by the intersection of the hyperplane $\mathcal{H}^\lambda$ and $\mathcal{L}$, and is then not considered in $\mathcal{L}$. Finally, the new LBS consists of the primal points retrieved from the ILP solver and the intersecting points as well.

\autoref{algo:blackbox_LBS_cut} repeats the above convex intersection procedure until no more potential searching direction $\lambda$ is available. The exploring direction $\lambda$ is enumerated in a dichotomic manner, starting from the two leftmost and rightmost points. For each new primal point $y^\lambda$ integrated into $\mathcal{L}$, the next direction parameters are the vectors respectively orthogonal to the line segments defined by $y^\lambda$ and its two adjacent points.

\subsection{Reducing the lower bound computation time}\label{sec:cuts_aggregation}

In the previous algorithms, the BOLP relaxation is computed exhaustively by enumerating all non-dominated extreme points on polyhedron $\widetilde{\mathcal{Y}}$ in \autoref{algo:B&B} and \autoref{algo:MP_cutting_plane}, and all potential optimizing directions are explored in \autoref{algo:blackbox_LBS_cut}. However for large problems, the cardinality of the extreme points on the non-dominated BOLP boundaries $\widetilde{\mathcal{Y}}_N$ is not polynomially bounded respected to the input instance size. To avoid that LBS construction and reinforcement overconsume computation time, a non-exact LBS calculation will be  parameterized with $\Lambda$, a small number of scalarization vectors, is suggested in the following.

\paragraph{Algorithm BOBLB\&C ($|\Lambda|$)} Since the LBS of the predecessor remains valid for the current node, the convex intersection of the local non-exact LBS with the predecessor's is still valid for the current subproblem by definition in \cite{ehrgott2007bound}. In \autoref{algo:blackbox_LBS_cut}, the size of $\Lambda$ and the scalar vectors $\lambda \in \Lambda$ may be defined by different strategies. By default, as mentioned in \autoref{algo:blackbox_LBS_cut}, parameters $\Lambda$ are determined in a dichotomic manner. $\Lambda$ can be settled equilibrately $\{\frac{1}{|\Lambda|}, \frac{2}{|\Lambda|}, \ldots , \frac{|\Lambda|-1}{|\Lambda|}\}$, which is most used in the literature. We also proposed the "chordal" strategy where $\Lambda$ are perpendicularly fixed by the $|\Lambda|$ equidistant segments connected by two non-adjacent lower bounds in $\mathcal{L}$.

\paragraph{Algorithm Cut\&Branch} Another common strategy in the ILP context is to put more cutting efforts at the root node to start with a good relaxation bound, and do not cut but only branch in the rest of the nodes for a fast enumeration. In this strategy, cutting the LBS is applied until all lower bounds are integer at the root node, and only the BOLP relaxation problem is solved in the rest of the tree. In bi-objective context, this strategy is equivalent to the two-phase method (see \cite{ulungu1995two, przybylski2010two}) without parallelism.

\section{Experimental results}\label{sec:experiments}

\begin{table}
\centering
\scriptsize
\begin{tabular}{lcccc}
\toprule

\textbf{Abbreviation} &  Extended Pareto Branching (EPB) &  Multi-Point cuts (MP) & ILP Solver Cuts (ISC) & $|\Lambda|$\\
\midrule

\textbf{B\&B} & No & No & No & No\\
\textbf{EPB B\&B} & Yes & No & No & No\\
\textbf{B\&C (MP)} & No & Yes & No& No \\
\textbf{EPB B\&C (MP)} & Yes & Yes & No & No\\


\textbf{B\&C (ISC)} & No& No & Yes& No \\
\textbf{EPB B\&C (ISC)} & Yes & No & Yes & No\\
\textbf{B\&C (ISC+MP)} & No & Yes & Yes& No \\
\textbf{EPB B\&C (ISC+MP)} & Yes & Yes & Yes& No \\

\textbf{B\&C (ISC)($|\Lambda|$)} & No& No & Yes  & Yes\\
\textbf{EPB B\&C (ISC)($|\Lambda|$)} & Yes & No & Yes  & Yes\\
\textbf{B\&C (ISC+MP)($|\Lambda|$)} & No & Yes & Yes  & Yes\\
\textbf{EPB B\&C (ISC+MP)($|\Lambda|$)} & Yes & Yes & Yes  & Yes\\
\textbf{Cut\&Branch} & No & root possible & root only & No \\

\bottomrule
\end{tabular}
\caption{The BOBLB\&B algorithm with different implemented functions.}
\label{tab:abbrev}
\end{table}

In this section, different cutting approaches are compared for our BOBLB\&C algorithm. \autoref{tab:abbrev} displays the abbreviations of our BOBLB\&B(\&C) algorithms with different strategies applied. For each version of the algorithm, we indicate if we apply 
\begin{itemize}
    \item[-] EPB : the Extended Pareto Branching
    \item[-] MP : the Multi-Point cutting planes of \autoref{algo:MP_cutting_plane}. Recall that we only propose cover inequality separators in this work, and accordingly we apply the MP algorithm only on knapsack-like instances.
    \item[-] ISC : the ILP Solver's internal Cuts. In \autoref{algo:blackbox_LBS_cut}, for each scalarization problem, the ILP solver optimizes the integer problem with the number of explored nodes limited to $0$. In our experiments, IBM ILOG \textsc{Cplex} Optimizer \cite{cplex} solver of version 22.1.0 is used to benefit the embedded cuts generators, heuristics, preprocessing and other internal functions. The best fractional primal solution found by ILP solver is extracted via callback functions. If both ISC and MP cuts are activated, the multi-point cutting planes algorithm is applied after obtaining the LBS cut by ILP solver. Moreover, we also retrieve the heuristic solutions found by ILP solver to improve the global UBS. 

    \item[-] $|\Lambda|$ : the maximum number of lower bounds computed on every node. An exhaustive LBS is calculated only at the root node. For the rest of the nodes, $|\Lambda|$ lower bounds are computed and the respective supporting planes are convexly intersected with the predecessor's LBS in \autoref{algo:LBS_update_intersect}.
\end{itemize}

\subsection{BOBLB\&C with multi-point cutting planes}\label{sec:exp_multi_cuts}
In this section, we analyze the impact of our multi-point separation technique. 

In \autoref{algo:MP_cutting_plane}, the single-point cut generator is called on each lower bound if no valid multi-point cut is found for the current LBS. To avoid redundant inequalities, a local cut pool is managed for each node so that a previously stored inequality which is still valid for the current LBS is added before calling cut generators. Whereas, note that there may still exist, due to the local nature of the cut pool, redundant cuts in terms of the whole BOBLB\&C tree.

\begin{table}
\scriptsize
\centering
\begin{tabular}{lrrrrrr}
\toprule
~ & ~ & ~ & \multicolumn{2}{c}{\textbf{B\&C (MP)}}  & \multicolumn{2}{c}{\textbf{EPB B\&C (MP)}}
\\
\cmidrule(r){4-5} \cmidrule(r){6-7} 
\textbf{Instance} & \textbf{n} & \textbf{m} & \textbf{\#SP cuts} &\textbf{\#MP cuts} & \textbf{\#SP cuts} &\textbf{\#MP cuts} \\
\midrule

\multirow{3}{*}{BOBKP} & 10-20 & 2.00 & 151.50 & 282.72 & 171.00 & 346.14 \\
~ & 30-40 & 2.00 & 68140.07 & 303451.79 & 79326.43 & 351616.43 \\ 
~ & 50 & 2.00 & 194218.57 & 656873.86 & 145485.57 & 786934.29 \\ 
\hline

\multirow{2}{*}{BOMDMKP} & 10-20 & 2.86 & 3632.66 & 7617.53 & 4493.04 & 12566.72 \\  
~ & 30-40 & 4.84 & 192397.11 & 690735.13 & 227118.67 & 1.06277e6\\ 
\hline

\multirow{3}{*}{MOKP-A} & 10-20 & 1.00 & 1275.09 & 1568.58 & 658.73& 1628.59 \\ 
~ &25-30 & 1.00 & 12798.49 & 17379.47 & 6048.89 & 14817.86\\
~ & 35-40 & 1.00 & 56065.50 & 78351.55 & 27407.97 & 57351.03 \\
\hline

\multirow{5}{*}{MOKP-B} & 10-20 & 1.00 & 3969.79 & 5076.74 & 2002.78 & 5494.58 \\
~ & 30-40 & 1.00 & 246242.68 & 321845.53 & 103766.13 & 274638.38 \\ 
~ & 50-60 & 1.00 & 651634.40 & 1.07692e6 & 376897.85 &1.01649e6 \\ 
~ & 70-80 & 1.00 & 537040.45 & 919172.85 & 384578.90 & 998296.05 \\ 
~ & 90-100 & 1.00 & 451224.25 & 771426.65 & 332495.20 & 1.00591e6 \\ 

\bottomrule
\end{tabular}
\caption{Comparison of the number of Single-Point(\#SP) cuts and Multi-Point(\#MP) cuts.}
\label{tab:BOBKP_BCMCvsEPBMC}
\end{table}

\begin{table}
\scriptsize
\centering
\scalebox{0.8}{
\begin{tabular}{lrrrrrrrrrrrrrr}
\toprule
~ & ~ & ~ & \multicolumn{3}{c}{\textbf{B\&B}} & \multicolumn{3}{c}{\textbf{B\&C(MP)}} & \multicolumn{3}{c}{\textbf{EPB B\&B}} & \multicolumn{3}{c}{\textbf{EPB B\&C(MP)}} \\

\cmidrule(r){4-6} \cmidrule(r){7-9} \cmidrule(r){10-12} \cmidrule(r){13-15}  
\textbf{Instance} & \textbf{n} & \textbf{m} & \textbf{Time(s)} & \textbf{\#TL} &\textbf{\#Nodes} & \textbf{Time(s)} & \textbf{\#TL} &\textbf{\#Nodes} & \textbf{Time(s)} & \textbf{\#TL} &\textbf{\#Nodes} & \textbf{Time(s)} & \textbf{\#TL} &\textbf{\#Nodes} \\
\midrule

\multirow{3}{*}{BOBKP} & 10-20 & 2.00 & \textbf{0.30} & 0/14 & 509.72 & 0.72 & 0/14 & 103.29  & 0.96 & 0/14 & 740.29 & 0.67  & 0/14 & 343.65 \\ 
~ & 30-40 & 2.00 &\textbf{375.42} & 1/14 & 109992.57 &419.87 & 1/14  & 20299.86 & 409.18 & 1/14 & 256052.86 & 383.16 & 1/14 & 170421.00 \\ 
~ & 50 & 2.00 & 1192.92 & 2/7 &165915.00 & 1217.61 & 2/7 & 38089.00 & 1359.82 & 2/7 & 386796.86 & \textbf{1066.23} & 1/7 & 217257.00\\ 
\hline

\multirow{2}{*}{BOMDMKP} & 10-20 & 2.86 & \textbf{12.71} & 0/21 & 15700.91 & 14.80 & 0/21 & 3048.72 & 13.58 & 0/21 & 16071.00 & 15.64 & 0/21 & 6478.90 \\
~ & 30-40 & 4.84 & \textbf{2503.30} & 32/57 & 648192.34 & 2749.22 & 38/57 & 152750.65 & 2598.29 & 34/57& 822636.19 & 2604.32 & 35/24 & 331586.38 \\
\hline

\multirow{3}{*}{MOKP-A} & 10-20 & 1.00 & \textbf{2.90} & 0/1080 & 1506.12 & 4.69 & 0/1080 & 613.26  & 4.77 & 0/1080 & 3738.68 & 5.19 & 0/1080 & 2554.76 \\
~ & 25-30 & 1.00 & \textbf{19.21} & 0/60 & 8402.90 & 39.80 & 0/60 & 3865.80 & 35.71 & 0/60& 42977.72 & 37.26 & 0/60 & 31487.07 \\
~ & 35-40 & 1.00 & \textbf{96.33} & 0/60 & 27932.67 & 205.75 & 0/60 & 13048.94 & 200.66 & 0/60& 210114.13 & 203.42 & 0/60& 164504.17 \\ 
\hline

\multirow{5}{*}{MOKP-B}& 10-20 &1.00& \textbf{6.91} & 0/60 & 6516.04 & 9.64 & 0/60 & 1549.04 & 7.59 & 0/60 & 6733.68 & 7.98 & 0/60 & 2403.27 \\
~ & 30-40 &1.00 & \textbf{354.98} & 0/40 & 181982.95 & 720.45 & 0/40 & 66776.05 & 496.96 & 0/40 & 490901.95 & 500.97 & 0/40 & 297102.05 \\ 
~ & 50-60 &1.00 & \textbf{2772.78} & 9/20 & 547773.50 & 3370.41 & 15/20 & 121161.70 & 3111.45 & 15/20 & 1.93375e6 & 3141.85 & 15/20 & 1.33261e6 \\
~ & 70-80 &1.00 & 3600.99 & 20/20 & 241904.20 & 3601.66 & 20/20 & 55104.40 & 3601.47 & 20/20 & 1.71603e6 & 3601.72 & 20/20 & 1.21872e6 \\
~ & 90-100 &1.00 & 3600.95 & 20/20 & 143654.90 & 3601.75 & 20/20 & 33215.80 & 3600.93 & 20/20 & 1.53496e6 & 3601.78& 20/20 & 950149.80 \\

\bottomrule

\end{tabular}
}
\caption{Efficiency of applying Multi-point (MP) cutting planes.}
\label{tab:BOBKP_BBvsBCMC}
\end{table}

\autoref{tab:BOBKP_BCMCvsEPBMC} displays the average number of single-point cuts and multi-point cuts respectively found during the BOBLB\&C algorithm with and without EPB. The multi-point cutting plane algorithm is indeed efficient since, in \autoref{tab:BOBKP_BCMCvsEPBMC}, more multi-point valid cuts are produced than single-point cuts for all knapsack-like instances. 

\autoref{tab:BOBKP_BBvsBCMC} compares the total time and the number of explored nodes during the BOBLB\&B tree with and without \autoref{algo:MP_cutting_plane}, where the best execution times are in bold. Thanks to multi-point cuts, around 80\% of nodes are reduced in (EPB) BOBLB\&B tree and for BOBKP instances, up to 95\% of the nodes are cut down. However, the MP cutting plane algorithm improves the computation time only for a few BOBKP instances, due to the highly time-consuming of reoptimizing the BOLP problem. Without EPB, the MP cutting plane algorithm retards almost two times the BOBLB\&B tree process despite the 80\% fewer nodes. With EPB, using MP cuts the BOBLB\&C algorithm's time performance is close to the EPB BOBLB\&B algorithm.

Despite of this time consuming raising, this multi-point cut process appears to be really efficient. Consequently, it can happen that for B\&B algorithms dedicated to particular problems, this technique can produce more successful results if, for instance, more tighter cuts generators are implemented. Moreover, in order to benefit from this better lower bound set, some dedicated heuristics or metaheuristics should be helpful to produce a better upper bound set and then to discard redundant searching spaces, especially for EPB strategy.


In \autoref{sec:Lambda_Cutbranch}, we work on some other insights to reduce the time spent on relaxation for our generic B\&C methods.


\subsection{BOBLB\&C with ILP solver's internal cuts  }

\begin{table}
\centering
\scriptsize
\scalebox{0.9}{
\begin{tabular}{lrrrrrrrrrrr}
\toprule
~ & ~ & ~ & \multicolumn{3}{c}{\textbf{B\&C (MP)}} & \multicolumn{3}{c}{\textbf{B\&C (ISC)}} & \multicolumn{3}{c}{\textbf{B\&C (ISC+MP)}} \\

\cmidrule(r){4-6} \cmidrule(r){7-9} \cmidrule(r){10-12}
\textbf{Instance} & \textbf{n} & \textbf{m}  & \textbf{Time(s)} & \textbf{\#TL} &\textbf{\#Nodes} & \textbf{Time(s)} & \textbf{\#TL} &\textbf{\#Nodes} & \textbf{Time(s)} & \textbf{\#TL} &\textbf{\#Nodes} \\
\midrule

\multirow{3}{*}{BOBKP} & 10-20 & 2.00 & 0.72 & 0/14 & 103.29 & 0.52  & 0/14 & 77.43 & 0.70 & 0/14 & 76.29 \\ 
~ & 30-40 & 2.00 & 419.87 & 1/14 & 20299.86 & 57.79 & 0/14 & 3474.43 & 73.74  & 0/14 & 3148.29 \\
~ & 50 & 2.00 & 1217.61 & 2/7 & 38089.00 &202.95 & 0/7& 7237.86 & 280.93 & 0/7  & 6651.29  \\ 
\hline

\multirow{2}{*}{BOMDMKP} &  10-20 & 2.86 & 14.80 & 0/21 & 3048.72 & 6.96 & 0/21 & 236.14 & 7.71 & 0/21 & 236.14 \\
~  & 30-40 & 4.84 & 2749.22 & 38/57 & 152750.65 & \textbf{1225.50} & 8/57 & 11544.65 & 1260.86 & 9/57 & 10912.68 \\
\hline

\multirow{3}{*}{MOKP-A} & 10-20 & 1.00 & 4.69 & 0/1080 & 613.26 & 5.27 & 0/1080 & 392.28 & 5.95 & 0/1080 & 386.99 \\
 ~ & 25-30 & 1.00 & 39.80 & 0/60 & 3865.80 & 28.79 & 0/60 & 2390.57 & 32.14 & 0/60 & 2310.37 \\
~& 35-40 & 1.00 & 205.75 & 0/60 & 13048.94 & 130.33 & 0/60 & 7595.87 & 151.58 & 0/60 & 7208.40 \\
\hline

\multirow{5}{*}{MOKP-B} & 10-20& 1.00 & 9.64 & 0/60 & 1549.04 & 3.77 & 0/60 & 175.57 & 4.46& 0/60 & 170.10 \\
~ & 30-40 & 1.00 & 720.45 & 0/40 & 66776.05 & 27.44 & 0/40 & 2113.85 & 33.65 & 0/40 & 2008.40 \\
~ & 50-60 &1.00 & 3370.41 & 15/20 & 121161.70 & 186.52& 0/20 & 9412.50 & 379.66& 0/20 & 8628.80 \\
~ &70-80 &1.00 & 3601.66 & 20/20 & 55104.40 & 1400.29 & 3/20 & 24852.90 &1443.19 & 3/20 & 19064.00 \\
~ & 90-100 &1.00& 3601.75 & 20/20 & 33215.80 &2666.82 & 7/20 & 33217.50 &3603.83 & 7/20 &25968.60 \\

\bottomrule

\multicolumn{12}{c}{\vspace{\baselineskip}}\\

\toprule
~ & ~ & ~ & \multicolumn{3}{c}{\textbf{EPB B\&C (MP)}} & \multicolumn{3}{c}{\textbf{EPB B\&C (ISC)}} & \multicolumn{3}{c}{\textbf{EPB B\&C (ISC+MP)}} \\

\cmidrule(r){4-6} \cmidrule(r){7-9} \cmidrule(r){10-12}
\textbf{Instance} & \textbf{n} & \textbf{m}  & \textbf{Time(s)} & \textbf{\#TL} &\textbf{\#Nodes} & \textbf{Time(s)} & \textbf{\#TL} &\textbf{\#Nodes} & \textbf{Time(s)} & \textbf{\#TL} &\textbf{\#Nodes} \\
\midrule

\multirow{3}{*}{BOBKP} & 10-20 & 2.00  & 0.67 & 0/14  & 343.65 & 0.36 & 0/14 & 58.29 & \textbf{0.29}  & 0/14 & 57.86\\ 
~ & 30-40 & 2.00  & 383.17  & 1/14 & 170421.00 & \textbf{10.19} & 0/14 & 1037.86 & 10.99 & 0/14 & 1045.15 \\
~ & 50 & 2.00 & 1066.20 & 1/7 & 217257.00 & \textbf{24.71} & 0/7 & 2689.14 & 25.92 & 0/7  & 2689.14 \\ 
\hline

\multirow{2}{*}{BOMDMKP} & 10-20 & 2.86 & 15.64 & 0/21 & 6478.90 & \textbf{6.89} & 0/21 & 228.67 & 7.14 & 0/21 & 228.53\\
~  & 30-40 & 4.84 & 2604.32 & 35/57 & 331586.38 & 1588.15 & 11/57 & 54725.67 &1594.15 & 11/57 & 54435.09\\
\hline

\multirow{3}{*}{MOKP-A} & 10-20 & 1.00 & 5.19 & 0/1080 & 2554.76 & \textbf{4.03} & 0/1080 & 134.25 & 4.27 & 0/1080 & 132.95\\ 
 ~ & 25-30 & 1.00 & 37.26 & 0/60 & 31487.07 & \textbf{10.28} & 0/60 & 1033.52 & 10.70 & 0/60 & 1030.53\\
~ & 35-40 & 1.00 & 203.42 & 0/60 & 164504.17 & \textbf{33.01} & 0/60 & 3732.67 & 33.93 & 0/60 & 3730.89 \\
\hline

\multirow{5}{*}{MOKP-B} & 10-20& 1.00& 7.98&  0/60  & 2403.27 & \textbf{3.62}&  0/60  & 87.32& 3.89 &  0/60 &85.79 \\
~& 30-40 & 1.00 & 500.97& 0/40 & 297102.05 & \textbf{11.17}& 0/40 &1000.13 & 11.99& 0/40 & 1011.15 \\
~ & 50-60& 1.00 & 3141.85 & 15/20 & 1.332615e6 & \textbf{54.69}& 0/20 & 5250.40 &  56.84 & 0/20 & 5245.40 \\
~& 70-80 &1.00&  3601.72 & 20/20 & 1.21872e6 & \textbf{382.46} & 0/20 & 31224.50 &  391.81 & 0/20 & 31235.65 \\
~ & 90-100 & 1.00 & 3601.78 & 20/20 & 950134.80 & \textbf{1109.48} & 0/20 & 79893.90 &  1158.20 & 0/20 & 81802.40 \\

\bottomrule
\end{tabular}
}
\caption{The BOBLB\&C performance with different cutting approaches.}
\label{tab:BOBKP_MCvsCPLEX}
\end{table}

\begin{table}
\centering
\scriptsize
\begin{tabular}{lrrrrrrrrrrr}
\toprule
~ & ~ &~  & \multicolumn{3}{c}{\textbf{B\&B}} & \multicolumn{3}{c}{\textbf{B\&C (ISC)}} & \multicolumn{3}{c}{\textbf{EPB B\&C (ISC)}} \\

\cmidrule(r){4-6} \cmidrule(r){7-9} \cmidrule(r){10-12}
\textbf{Instance} & \textbf{n} & \textbf{m} & \textbf{Time(s)} & \textbf{\#TL} &\textbf{\#Nodes} & \textbf{Time(s)} & \#TL &\textbf{\#Nodes} & \textbf{Time(s)} & \textbf{\#TL} &\textbf{\#Nodes} \\
\midrule

\multirow{3}{*}{BOBSCP} & 20-40 & 6.00 & \textbf{2.10} & 0/24 & 374.42 & 4.21 & 0/24 & 1431.83 & 3.53 & 0/24 & 170.50\\ 
~ & 60-80 & 13.44  & 16.26 & 0/27 & 3144.93 & 40.71 & 0/27 & 3468.18 & \textbf{14.16} & 0/27 & 767.11\\
~ & 100 & 22.00 & \textbf{114.47} & 0/15 & 9436.73 & 138.85 & 0/15 &6878.73 & 138.07 & 0/15 & 3697.93\\
\hline

\multirow{3}{*}{BOSPA} & 50-500 & 21.17 & \textbf{15.50} & 0/6& 3114.33 & 65.10& 0/6 & 4361.00 & 18.26 & 0/6 & 933.67  \\
~ & 550-1500 & 21.94& 249.19 & 0/16 & 16997.62 & 568.94 & 1/16 & 14583.00 & \textbf{175.17} & 0/16 & 3988.75 \\
~ & 1600-3000 & 25.50&  1505.24 & 1/6 & 28945.67 & 2107.96 & 1/6 & 28259.67 & \textbf{975.70} & 0/6 & 6739.17\\
\hline

\multirow{2}{*}{MOAP} & 100-225 & 25.00  & \textbf{60.94} & 0/20 & 12436.50 & 305.13 & 0/20 & 18476.10 & 72.51 & 0/20 & 2677.00 \\ 
~ & 400-625 & 45.00 & 2262.61 & 9/20 & 104902.20 & 3532.87 & 17/20 & 39593.20 & \textbf{1868.96} & 5/20 & 22688.75 \\
\hline

\multirow{4}{*}{MOUFLP} & 30-42 & 41.08 & \textbf{4.11} & 0/130& 2189.31 & 15.21 & 0/130 & 3760.60 & 9.51 &0/130 & 654.52 \\
~ & 56-72& 57.23 & \textbf{11.48} & 0/130 & 5169.66 & 52.47 & 0/130 & 8858.64 & 54.26 & 0/130 & 2047.49 \\
~ & 110-156 & 133.00 & \textbf{9.61} & 0/240 & 2163.35 & 42.44 & 0/240 & 2602.98 & 17.69& 0/240 & 672.44\\
~ & 210& 210.00 &  \textbf{50.17}& 0/120 & 8181.37& 232.08 & 0/120 &9976.65 &101.60& 0/120 & 2268.18 \\

\bottomrule
\end{tabular}
\caption{The BOBLB\&C performance with ILP solver oracle cutting approach.}
\label{tab:BO_BBvsCPLEX}
\end{table}

\autoref{tab:BOBKP_MCvsCPLEX} respectively displays the average computation time, the number of instances unsolved within time limit (TL) and the explored tree size of all versions of BOBLB\&C algorithms implemented in this paper. This table focuses on knapsack-like instances to compare those methods including the multi-point cutting planes dedicated to these instances. 

Overall, the ILP solver oracle cutting (ISC) approaches significantly improve both the computation time and the number of pruned nodes, compared to the multi-point cutting algorithm. Simultaneously the ISC and MP cuts reduce the most of the B\&C tree size. 

Whereas using only MP cutting planes, as done in the previous section, the EPB strategy was not helpful; the EPB strategy with ISC cuts immensely improves the BOBLB\&C performance both in time and on the tree size and therefore outperforms the other methods. The ISC strategy not only efficiently reinforces the LBS but also improves the UBS with internal heuristics. 

Finally, the most sophisticated version with EPB, ISC and MP strategies leads to the least number of nodes but slightly more time consuming, in spite of the good quality of the LBS where the lower bounds are close to integer points.

\autoref{tab:BO_BBvsCPLEX} presents the results obtained with EPB and ISC cuts for the other instances. Our method significantly prunes nodes by reinforcing the bound sets. The number of nodes is always prominently reduced, but sometimes the total time is not accelerated. This is particularly the case for the MOUFLP instances for which the number of nodes is not sufficiently reduced compared to other instances. More tight inequalities for special problems should be usefully integrated in a dedicated approach for particular problems, like MOUFLP instances.

\subsection{Reduction of the lower bound set computation time}\label{sec:Lambda_Cutbranch}

\begin{figure}[!ht]
\centering
\scalebox{0.45}{\includegraphics{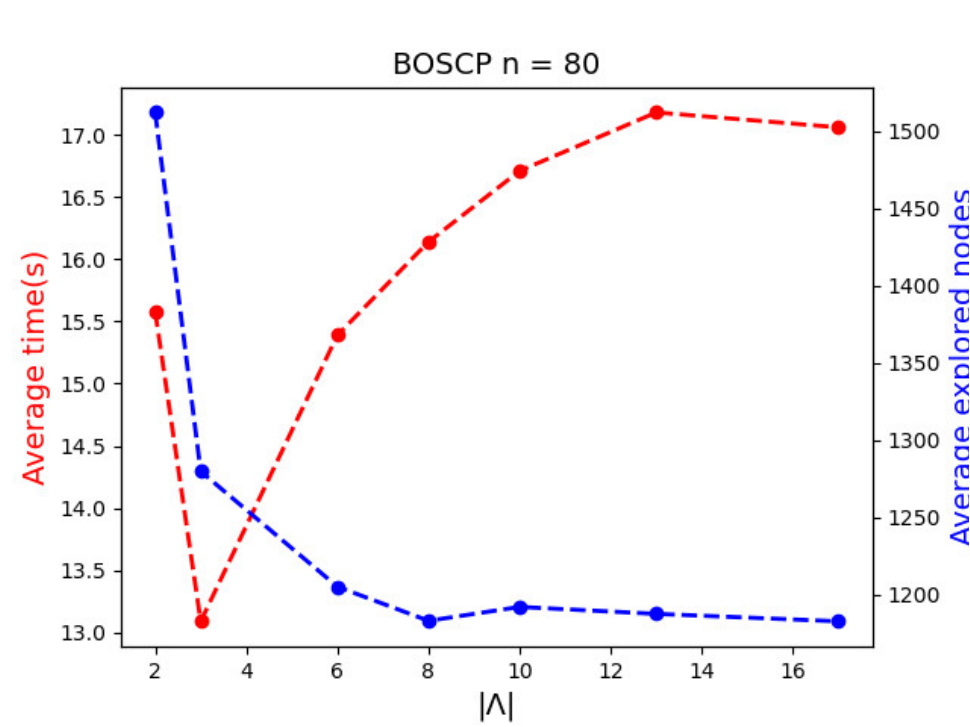}}
\hfill
\scalebox{0.45}{ \includegraphics{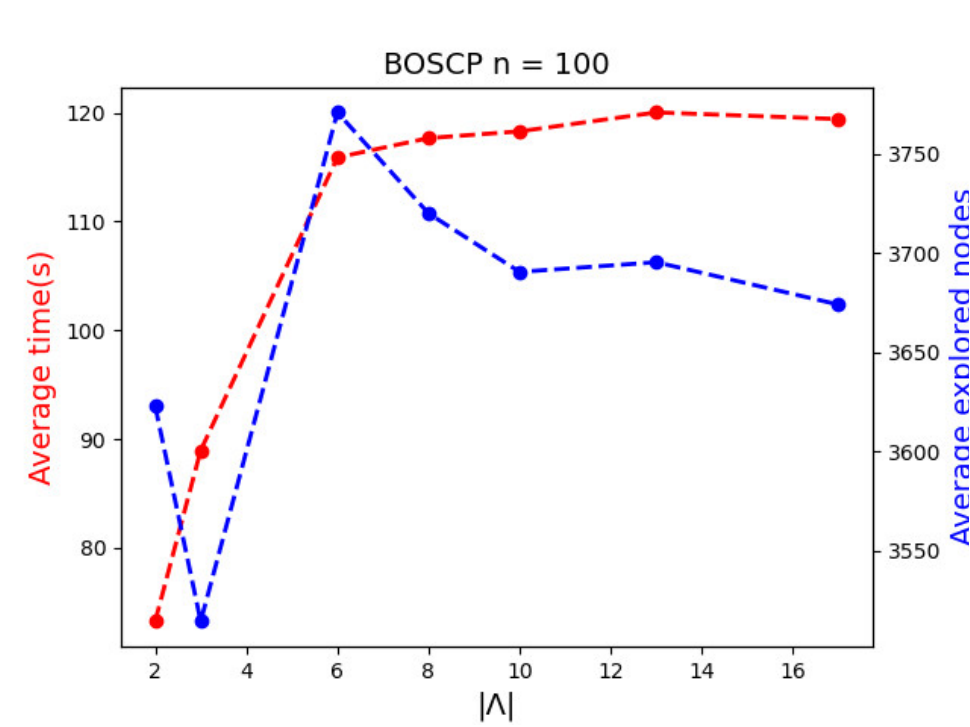}}
\caption{The BOBLB\&C with EPB and ISC performance evolution according to the number of $\Lambda$ parameters limited on each node for BOSCP instances .}
\label{fig:BOSCP_EPB_BOBC_lambda}
\end{figure}

Previous numerical results show that (re)optimizing the LBS is highly time-consuming and occupies almost all of the total computation time. Rather than exhaustively computing all potential non-dominated extreme lower bounds, \autoref{sec:cuts_aggregation} proposes to calculate a partial LBS to intersect with the predecessor's LBS. To study the trade-off between the exactness of LBS and algorithm performance, \autoref{fig:BOSCP_EPB_BOBC_lambda} illustrates our computational test to precise 
the number $|\Lambda|$ of single-objective scalarization
problems solved at each node. Each graphic plots the computation time and BOBLB\&C tree size behaviors according to $|\Lambda|$. The total consuming time drops dramatically with small $|\Lambda|$ size, while the BOBLB\&C algorithm completion time sharply increases with large $|\Lambda|$. We aim to achieve the tradeoff between the time spent on LBS calculation and the number of nodes pruned. It is straightforward that the more LBS is completed, the fewer nodes are generated by BOBLB\&C algorithm and accordingly the more time is spent. In conclusion, the appropriate limit $|\Lambda|$ is around 2 and 3.

\begin{table}[!ht]
\centering
\scriptsize
\scalebox{0.8}{
\begin{tabular}{lrrrrrrrrrrrr}
\toprule
~& ~& ~ & \textbf{\tiny{$\epsilon$-constraint}}  & \multicolumn{3}{c}{\textbf{B\&B}} & \multicolumn{3}{c}{\textbf{EPB B\&C (ISC) ($|\Lambda|$)}} &  \multicolumn{3}{c}{\textbf{Cut\&Branch}} 
\\
\cmidrule(r){4-4} \cmidrule(r){5-7} \cmidrule(r){8-10} \cmidrule(r){11-13} 
\textbf{Instance} & \textbf{n} & \textbf{m}  &\textbf{Time(s)}  & \textbf{Time(s)} & \textbf{\#TL} &\textbf{\#Nodes} & \textbf{Time(s)} & \textbf{\#TL} &\textbf{\#Nodes}& \textbf{Time(s)} & \textbf{\#TL} &\textbf{\#Nodes} \\
\midrule

\multirow{3}{*}{BOBKP} & 10-20 & 2.00  & 0.99 & 0.30 & 0/7 & 509.72 & \underline{\textbf{0.20}} & 0/14 & 58.29 & 0.71 & 0/14 & 8.43 \\
~ & 30-40 & 2.00 & \textbf{5.37} &375.42 & 0/7 & 109992.57 & \underline{7.31} & 0/14 & 1037.86 & 141.84 & 0/14 & 27.72 \\ 
~ & 50 & 2.00 & \textbf{10.74} & 1192.92 & 2/7 & 165915.00 & \underline{17.97} & 0/7 & 2689.14& 568.97& 1/7&40.29 \\
\hline

\multirow{2}{*}{BOMDMKP}  & 10-20 & 2.86 & 4.88 & 12.71 & 0/21 & 15700.91  & \underline{\textbf{2.69}} & 0/21 & 327.57 & 5.27 & 0/21 & 2934.34 \\ 
~ & 30-40 & 4.84 & \textbf{13.74} & 2503.30 & 32/57 & 648192.34 & \underline{594.98} & 2/57 & 27300.40 & 1414.46 & 10/57 & 326980.12 \\
\hline

\multirow{3}{*}{MOKP-A} & 10-20& 1.00 &4.18& 2.90 & 0/1080 & 1506.12 & \underline{\textbf{0.69}} & 0/1080 & 134.25 & 5.27 & 0/1080 & 392.28 \\
~& 25-30& 1.00& 7.27&  19.21 & 0/60 & 8402.90 & \underline{\textbf{6.64}} & 0/60 & 1033.52 & 28.79 & 0/60 & 2390.57  \\
~& 35-40& 1.00& \textbf{12.63} & 96.33 & 0/60  & 27932.67 & \underline{29.23} & 0/60 & 3732.67 & 130.34 & 0/60 & 7595.87  \\
\hline

\multirow{5}{*}{MOKP-B} & 10-20 & 1.00& 3.84& 6.91& 0/60 & 6516.04  & \underline{\textbf{0.43}} & 0/60 & 87.32 & 3.48 & 0/60 & 917.77  \\
~ & 30-40& 1.00 &\textbf{5.91} &354.98& 0/40 & 181982.95 & \underline{7.64} & 0/40 & 1000.13 & 59.63 & 0/40 & 27929.50  \\
~ &50-60& 1.00& \textbf{10.01}& 2772.78& 9/20 & 547773.50  & \underline{50.93} & 0/20 & 5250.40 & 1016.39 & 2/20 & 210694.00 \\
~ & 70-80 &1.00& \textbf{16.93} &  3600.98 & 20/20 & 241904.20 & \underline{380.77} & 0/20 & 31224.35 & 2971.46 & 11/20 & 319763.30 \\
~ & 90-100& 1.00& \textbf{24.21} &  3600.95 & 20/20 & 143654.90 & \underline{1107.85} & 0/20 & 79892.90 & 3598.72 & 19/20 & 184308.10 \\
\hline

\multirow{3}{*}{BOSCP} & 20-40 & 6.00 & 4.01 & 2.10 & 0/24 & 374.42 & \underline{\textbf{0.72}} & 0/24 & 202.29 & 2.61 & 0/24 & 356.50  \\
~ & 60-80 & 13.44 & \textbf{6.92} & 16.26 & 0/27 & 3144.93 & \underline{7.91} & 0/27 & 917.59 & 14.75 & 0/27 & 2837.23 \\
~ & 100 & 22.00 & \textbf{9.74} &114.47 & 0/15 & 9436.73 & \underline{73.25}& 0/15 & 3623.07& 82.75& 0/15 & 7382.87\\
\hline

\multirow{3}{*}{BOSPA} & 50-500& 21.94 & \textbf{4.91} &15.50 & 0/6 & 3114.33 & \underline{14.45} & 0/6 & 930.33 & 15.31 & 0/6 & 2871.67  \\
~ & 550-1500& 22.70 & \textbf{6.64} & 249.19 & 0/16 & 16997.62 & \underline{160.30} & 0/16 & 3757.08 & 211.17 & 0/16 & 15374.87  \\
~ & 1600-3000 & 25.50  &\textbf{8.41} &  1505.24 & 1/6 & 28945.67 & \underline{944.77} & 0/6 & 6744.00 & 1292.68 & 1/6 & 22169.33 \\
\hline

\multirow{2}{*}{MOAP} & 100-225 & 25.00 & \textbf{6.66} &  \underline{60.94} & 0/20 &12436.50 & 68.14 & 0/20 & 2677.00 & 66.05 & 0/20 & 12436.50 \\
~ & 400-625 & 45.00 & \textbf{11.72} &2262.61 & 9/20 & 104902.20 & \underline{1874.86} & 1/20 & 22380.60 & 2334.42 & 10/20 & 91786.80 \\
\hline

\multirow{4}{*}{MOUFLP} & 30-42 & 41.08& 6.63 & \underline{\textbf{4.11}} & 0/130 & 2189.31  & 6.39 & 0/130 & 654.52 & 5.04 & 0/130 & 2189.14  \\
~ & 56-72 & 57.23 & \textbf{9.68} & \underline{11.48} & 0/130 & 5169.66 & 76.90 & 0/130 & 2128.97 & 12.98 & 0/130 & 5170.91 \\
~ & 110-156& 133.00& \textbf{4.52} & \underline{9.61} & 0/240 & 2163.35 & 14.93 & 0/240 & 672.44 & 11.34 & 0/240 & 2161.94 \\
~ & 210& 210.00& \textbf{7.03} & \underline{50.17}& 0/120 & 8181.37 & 98.27& 0/120 & 2268.18  & 56.27 & 0/120&8175.92 \\
\bottomrule

\end{tabular}
}
\caption{A performance comparison between the $\epsilon$-constraint and BOBLB\&B\&C algorithms.}
\label{tab:BOBKP_EPBBCvsEpsilonCtr}
\end{table}

In \autoref{tab:BOBKP_EPBBCvsEpsilonCtr}, the EPB B\&C (ISC)($\Lambda$) column displays the EPB B\&C (ISC) algorithm with the parameter $|\Lambda|$ discussed above. In addition the last column presents the Cut\&Branch method which starts with the dichotomic method using an ILP solver at root node. 

\autoref{tab:BOBKP_EPBBCvsEpsilonCtr} compares our most sophisticated versions of BOBLB\&C algorithms with the pure BOBLB\&B and $\epsilon$-constraint algorithms. The best computation times are in bold and the fastest approach among BOBLB\&B\&C algorithms is underlined.

The Cut\&Branch strategy improves the B\&B methods but not as much as our BOBLB\&C algorithm. Note that this is the case for large MOUFLP instances when the BOBLB\&C algorithm is less competitive.

The BOBLB\&C algorithm outperforms the other methods, for knapsack-like instances and especially for those of reasonable size. For large scale instances, our algorithm significantly improves the simple BOBLB\&B approach with a factor of 100 of the computing time. For BOSCP and BOSPA instances, the BOBLB\&C algorithm reduces remarkably the explored nodes but saves less computation time. However for MOUFLP instances, the ILP solver's internal cuts still do not improve the BOBLB\&B algorithm. In future work, more tight specific valid inequalities could be integrated in next studies.


In conclusion, the $\epsilon$-constraint method globally outperforms for the large instances. However our BOBLB\&C algorithms immensely improve the B\&B algorithm for each case of instances. The complete experiments tables for all instances can be found at \url{https://github.com/Yue0925/vOptGeneric.jl/tree/master/tables}.




\section{Conclusion}\label{sec:conclusion}

In this paper, we devise two generic cutting approaches integrated into bi-objective branch\&cut algorithms, which both show the prominent reinforcement on the lower bound set compared to pure Branch\&Bound methods. Our multi-point valid inequalities are shown to dominate the classical single-point cuts. To strengthen the relaxation bound set, exploiting advanced techniques inside ILP solvers experimentally outperforms as well. 

The two generic approaches are easy to extend to many-objective context and for integer linear problems. Note that for many-objective ($p \geqslant 3$) combinatorial problems, the $\epsilon$-constraint algorithm starts to lose its advantage in criteria space enumeration, while the number of objective functions does not influence the difficulty of branch\&bound\&cut algorithms. One of the future works is to extend our proposed techniques in order to obtain a state-of-the-art method for $p$ greater than 2.

Our BOBLB\&C algorithms performance can still be improved by well-designed heuristics to enhance the upper bound set and the Pareto branching. Regarding reducing global time, BOBLB\&C algorithms could benefit from the parallelization that divides the criterion space into independent subproblems, such as the two-phase \cite{przybylski2010two} and slicing \cite{stidsen2018hybrid} methods.

Applying our multi-point separation idea for non-compact formulations -for instance for the traveling salesman problem- is worth studying to see the competence with other criterion space search approaches. A further application of our framework on integer nonlinear problems can also be interesting to study.

{\small\bf Acknowledgments:} We would like to express our sincere gratitude to Xavier Gandibleux and Anthony Przybylski for their contribution that initiates this research.

\bibliographystyle{plain}
\bibliography{refs}

\begin{thebibliography}{10}

\bibitem{adelgren2022branch}
Nathan Adelgren and Akshay Gupte.
\newblock Branch-and-bound for biobjective mixed-integer linear programming.
\newblock {\em INFORMS Journal on Computing}, 34(2):909--933, 2022.

\bibitem{aneja1979bicriteria}
Yash~P Aneja and Kunhiraman~PK Nair.
\newblock Bicriteria transportation problem.
\newblock {\em Management Science}, 25(1):73--78, 1979.

\bibitem{beasley1990or}
John~E Beasley.
\newblock Or-library: distributing test problems by electronic mail.
\newblock {\em Journal of the operational research society}, 41(11):1069--1072,
  1990.

\bibitem{ben2006constraint}
Walid Ben-Ameur and Jos{\'e} Neto.
\newblock A constraint generation algorithm for large scale linear programs
  using multiple-points separation.
\newblock {\em Mathematical programming}, 107(3):517--537, 2006.

\bibitem{boland2015criterion}
Natashia Boland, Hadi Charkhgard, and Martin Savelsbergh.
\newblock A criterion space search algorithm for biobjective integer
  programming: The balanced box method.
\newblock {\em INFORMS Journal on Computing}, 27(4):735--754, 2015.

\bibitem{10.1007/978-3-642-87563-2_5}
V.~Joseph Bowman.
\newblock On the relationship of the tchebycheff norm and the efficient
  frontier of multiple-criteria objectives.
\newblock In Herv{\'e} Thiriez and Stanley Zionts, editors, {\em Multiple
  Criteria Decision Making}, pages 76--86, 1976.

\bibitem{bokler2024outer}
Fritz Bökler, Sophie Parragh, Markus Sinnl, and Fabien Tricoire.
\newblock An outer approximation algorithm for generating the
  edgeworth–pareto hull of multi-objective mixed-integer linear programming
  problems.
\newblock {\em Mathematical Methods of Operations Research}, 100(1):1--28, 01
  2024.

\bibitem{cappanera2005local}
Paola Cappanera and Marco Trubian.
\newblock A local-search-based heuristic for the demand-constrained
  multidimensional knapsack problem.
\newblock {\em INFORMS Journal on Computing}, 17(1):82--98, 2005.

\bibitem{cerqueus2015bi}
Audrey Cerqueus.
\newblock {\em Bi-objective branch-and-cut algorithms applied to the binary
  knapsack problem}.
\newblock PhD thesis, universit{\'e} de Nantes, 2015.

\bibitem{chankong1983}
Vira Chankong and Yacov~Y Haimes.
\newblock Multiobjective decision making theory and methodology.
\newblock {\em Elsevier Science Publishing, New York}, 1983.

\bibitem{ehrgott2006discussion}
Matthias Ehrgott.
\newblock A discussion of scalarization techniques for multiple objective
  integer programming.
\newblock {\em Annals of Operations Research}, 147(1):343--360, 2006.

\bibitem{ehrgott2007bound}
Matthias Ehrgott and Xavier Gandibleux.
\newblock Bound sets for biobjective combinatorial optimization problems.
\newblock {\em Computers \& Operations Research}, 34(9):2674--2694, 2007.

\bibitem{florios2010solving}
Kostas Florios, George Mavrotas, and Danae Diakoulaki.
\newblock Solving multiobjective, multiconstraint knapsack problems using
  mathematical programming and evolutionary algorithms.
\newblock {\em European Journal of Operational Research}, 203(1):14--21, 2010.

\bibitem{forget2022warm}
Nicolas Forget, Sune~Lauth Gadegaard, and Lars~Relund Nielsen.
\newblock Warm-starting lower bound set computations for branch-and-bound
  algorithms for multi objective integer linear programs.
\newblock {\em European Journal of Operational Research}, 302(3):909--924,
  2022.

\bibitem{gadegaard2019bi}
Sune~Lauth Gadegaard, Lars~Relund Nielsen, and Matthias Ehrgott.
\newblock Bi-objective branch-and-cut algorithms based on lp relaxation and
  bound sets.
\newblock {\em INFORMS Journal on Computing}, 31(4):790--804, 2019.

\bibitem{gadegaard2020branch}
Sune~Lauth Gadegaard, Lars Relund, Nicolas~Joseph Forget, Kathrin Klamroth, and
  Anthony Przybylski.
\newblock Branch-and-bound and objective branching with three or more
  objectives.
\newblock {\em Computers \& Operations Research}, 148:106012, 2022.

\bibitem{gandibleux2021primal}
Xavier Gandibleux, Guillaume Gasnier, and Sa{\"\i}d Hanafi.
\newblock A primal heuristic to compute an upper bound set for multi-objective
  0-1 linear optimisation problems.
\newblock In {\em Proc. of the 1st Multi-Objective Decision Making Workshop
  (MODeM 2021)}, 2021.

\bibitem{vOptLib}
Xavier Gandibleux, Gauthier Soleilhac, and Anthony Przybylski.
\newblock voptlib: Library of numerical instances for multiobjective linear
  optimization problems, 2021.
\newblock \url{https://github.com/vOptSolver/vOptLib}.

\bibitem{vOptSolver}
Xavier Gandibleux, Gauthier Soleilhac, and Anthony Przybylski.
\newblock voptsolver: an ecosystem for multi-objective linear optimization.
\newblock JuliaCon 2021, July 28-30 2021.
\newblock \url{http://github.com/vOptSolver}.

\bibitem{cplex}
Inc IBM ILOG CPLEX~Optimization.
\newblock Cplex optimization studio reference manuals, 2021.
\newblock \url{https://www.ibm.com/analytics/cplex-optimizer}.

\bibitem{jozefowiez2012generic}
Nicolas Jozefowiez, Gilbert Laporte, and Fr{\'e}d{\'e}ric Semet.
\newblock A generic branch-and-cut algorithm for multiobjective optimization
  problems: Application to the multilabel traveling salesman problem.
\newblock {\em INFORMS Journal on Computing}, 24(4):554--564, 2012.

\bibitem{julia}
Inc JuliaHub.
\newblock The julia programming language, 2015.
\newblock \url{https://julialang.org}.

\bibitem{kirlik2014new}
Gokhan Kirlik and Serpil Say{\i}n.
\newblock A new algorithm for generating all nondominated solutions of
  multiobjective discrete optimization problems.
\newblock {\em European Journal of Operational Research}, 232(3):479--488,
  2014.

\bibitem{kiziltan1983algorithm}
G{\"u}lseren Kiziltan and Erkut Yucao{\u{g}}lu.
\newblock An algorithm for multiobjective zero-one linear programming.
\newblock {\em Management Science}, 29(12):1444--1453, 1983.

\bibitem{mavrotas2009effective}
George Mavrotas.
\newblock Effective implementation of the $\varepsilon$-constraint method in
  multi-objective mathematical programming problems.
\newblock {\em Applied mathematics and computation}, 213(2):455--465, 2009.

\bibitem{przybylski2017multi}
Anthony Przybylski and Xavier Gandibleux.
\newblock Multi-objective branch and bound.
\newblock {\em European Journal of Operational Research}, 260(3):856--872,
  2017.

\bibitem{przybylski2010two}
Anthony Przybylski, Xavier Gandibleux, and Matthias Ehrgott.
\newblock A two phase method for multi-objective integer programming and its
  application to the assignment problem with three objectives.
\newblock {\em Discrete Optimization}, 7(3):149--165, 2010.

\bibitem{ralphs2006improved}
Ted~K Ralphs, Matthew~J Saltzman, and Margaret~M Wiecek.
\newblock An improved algorithm for solving biobjective integer programs.
\newblock {\em Annals of Operations Research}, 147:43--70, 2006.

\bibitem{ramos1998problem}
Rosa~M Ramos, Sergio Alonso, Joaqu{\'\i}n Sicilia, and Carlos Gonz{\'a}lez.
\newblock The problem of the optimal biobjective spanning tree.
\newblock {\em European Journal of Operational Research}, 111(3):617--628,
  1998.

\bibitem{sourd2008multiobjective}
Francis Sourd and Olivier Spanjaard.
\newblock A multiobjective branch-and-bound framework: Application to the
  biobjective spanning tree problem.
\newblock {\em INFORMS Journal on Computing}, 20(3):472--484, 2008.

\bibitem{steuer1983interactive}
Ralph~E Steuer and Eng-Ung Choo.
\newblock An interactive weighted tchebycheff procedure for multiple objective
  programming.
\newblock {\em Mathematical programming}, 26:326--344, 1983.

\bibitem{stidsen2018hybrid}
Thomas Stidsen and Kim~Allan Andersen.
\newblock A hybrid approach for biobjective optimization.
\newblock {\em Discrete Optimization}, 28:89--114, 2018.

\bibitem{stidsen2014branch}
Thomas Stidsen, Kim~Allan Andersen, and Bernd Dammann.
\newblock A branch and bound algorithm for a class of biobjective mixed integer
  programs.
\newblock {\em Management Science}, 60(4):1009--1032, 2014.

\bibitem{ulungu1995two}
Ekunda~Lukata Ulungu and Jacques Teghem.
\newblock The two phases method: An efficient procedure to solve bi-objective
  combinatorial optimization problems.
\newblock {\em Foundations of computing and decision sciences}, 20(2):149--165,
  1995.

\end{thebibliography}

\end{document}